\documentclass[12pt]{amsart}
\usepackage[usenames,dvipsnames]{color}
\usepackage{amsmath,amsthm,amsfonts,amssymb,multicol,amscd,amsbsy}
\usepackage{graphicx}
\usepackage{booktabs}
\usepackage{subfigure}
\usepackage{enumerate}
\usepackage{xcolor}
\usepackage{url}
\usepackage[nodayofweek]{datetime}
\usepackage[english]{babel}
\usepackage[toc,page]{appendix}
\usepackage{verbatim} 
\usepackage{enumitem}
\usepackage{enumerate}
\usepackage{bbm} 
\usepackage[normalem]{ulem} 
\usepackage{bm}
\usepackage{hyperref}
\usepackage{amsmath}
\usepackage{nicefrac}

\usepackage{extarrows}
\usepackage[normalem]{ulem}

\usepackage{pgfplots}
\usepackage{caption}

\usepackage[top=2.75cm,bottom=3.75cm]{geometry}

\newcommand{\e}{\end}
\newcommand{\beq}{\begin{equation}}
\newcommand{\eeq}{\end{equation}}
\newcommand{\beqs}{\begin{equation*}}
\newcommand{\eeqs}{\end{equation*}}
\newcommand{\bal}{\begin{align}}
\newcommand{\eal}{\end{align}}
\newcommand{\bals}{\begin{align*}}
\newcommand{\eals}{\end{align*}}

\newcommand{\T}{\mathbb{T}}

\newcommand{\E}{\mathbb{E}}

\renewcommand{\Re}{\mathrm{Re}}
\renewcommand{\Im}{\mathrm{Im}}

\newcommand{\id}{\mathbbm {1}}

\def\ran{\operatorname{Ran}}

\def\dist{\operatorname{dist}}   
\def\dim{\operatorname{dim}}  
\def\tr{\operatorname{tr}}    

\newcommand{\set}[1]{\mathbb{#1}}

\newcommand{\R}{\set{R}}
\newcommand{\C}{\set{C}}
\newcommand{\Z}{\set{Z}}

\newcommand{\ip}[1]{\left\langle #1 \right \rangle}
\newcommand{\ipc}[2]{\left \langle #1 , \ #2 \right \rangle }
\newcommand{\Ev}[1]{\E \left( #1 \right)}  

\newcommand{\norm}[1]{\left\Vert#1\right\Vert}
\newcommand{\abs}[1]{\left\vert#1\right\vert}
\newcommand{\setb}[2]{\left \{ #1 \ \middle | \ #2 \right \} }
\newcommand{\com}[2]{\left[ #1 , #2 \right ]}

\newcommand{\Evac}[2]{\bb{E} \left ( #1 \middle | #2 \right )}
\newcommand{\Eva}[1]{\bb{E} \left ( #1  \right )}

\newcommand{\bb}[1]{\mathbb{#1}}
\newcommand{\mc}[1]{\mathcal{#1}}
\newcommand{\wt}[1]{\widetilde{#1}}
\newcommand{\wh}[1]{\widehat{#1}}
\renewcommand{\vec}[1]{\mathbf{#1}}
\newcommand{\wc}[1]{\wh{\mc{#1}}}
\newcommand{\ora}[1]{\overrightarrow{#1}}

\def\e{\mathrm e}

\def\im{\mathrm i}
\def\Im{\mathrm{Im}}

\def\1{{\mathsf 1}}
\def\di{\mathrm d}
\def\grad{\nabla}
\def\wt{\widetilde}

\newtheorem{thm}{Theorem}[section]
\newtheorem{lem}[thm]{Lemma}
\newtheorem{cor}[thm]{Corollary}
\newtheorem{prop}[thm]{Proposition}

\theoremstyle{definition}
\newtheorem{defn}[thm]{Definition}

\newtheorem{ass}[thm]{Assumption}

\newtheorem{conj}[thm]{Conjecture}

\numberwithin{equation}{section}

\theoremstyle{remark}
\newtheorem{remark}[thm]{Remark}

\def\dotuline{\bgroup
  \ifdim\ULdepth=\maxdimen  
   \settodepth\ULdepth{(j}\advance\ULdepth.4pt\fi
  \markoverwith{\begingroup
  \advance\ULdepth0.08ex
  \lower\ULdepth\hbox{\kern.15em .\kern.1em}%
  \endgroup}\ULon}

\def\dashuline{\bgroup
  \ifdim\ULdepth=\maxdimen  
   \settodepth\ULdepth{(j}\advance\ULdepth.4pt\fi
  \markoverwith{\kern.15em
  \vtop{\kern\ULdepth \hrule width .3em}%
  \kern.15em}\ULon}
\allowdisplaybreaks

\begin{document}

\setlength{\columnsep}{5pt}
\title[Diffusion for a Markovian periodic Schr\"odinger equation]{Diffusion in the Mean for a periodic Schr\"odinger Equation Perturbed by a Fluctuating Potential}
\author{Jeffrey Schenker}
\address{Mathematics Department \\ Michigan State University \\
	619 Red Cedar Road \\ East Lansing, MI 48823 }
\email{schenke6@msu.edu}
\author{F. Zak Tilocco}
\address{Mathematics Department \\ Michigan State University \\
	619 Red Cedar Road \\ East Lansing, MI 48823 }
\email{tiloccof@msu.edu}
\author{Shiwen Zhang}
\address{Mathematics Department \\ Michigan State University \\
	619 Red Cedar Road \\ East Lansing, MI 48823 }
\email{zhangshiwen@math.msu.edu}

\date{}

\begin{abstract}
We consider the evolution of a quantum particle hopping on a cubic lattice in any dimension and subject to a potential consisting of a periodic part and a random part that fluctuates stochastically in time. If the random potential evolves according to a stationary Markov process,  we obtain diffusive scaling for moments of the position displacement, with a diffusion constant that grows as the inverse square of the disorder strength at weak coupling. More generally,  we show that a central limit theorem holds such that the square amplitude of the wave packet converges, after diffusive rescaling, to a solution of a heat equation.
\end{abstract}

\maketitle

\section{Introduction and the Main Results}

Diffusive propagation is expected and observed to emerge from wave motion in a random medium in a variety of situations.  The general intuition behind this expectation is that repeated scattering from the random medium leads to a loss of coherence, which in a multi-scattering expansion or path integral formulation suggests a relation with random walks and diffusion.  This intuition is notoriously difficult to make precise in the context of a static random environment.  Indeed, proving the emergence of diffusion for the Schr\"odinger wave equation with a weakly disordered potential, in dimension $d\ge 3$, is one of the key outstanding open problems of mathematical physics.  For a random environment that fluctuates stochastically in time, the analysis is simpler and diffusive propagation has proved amenable to rigorous methods. Heuristically, this simplification is to be expected because time fluctuations suppress recurrence effects in path expansions. 

The present paper is the continuation of a project initiated by the first author and collaborators \cite{KS2009,HKS2010,MS2015,Schenker:2015,FS2015} in which diffusive propagation has been shown to occur for solutions to a tight binding Schr\"odinger equation with a random potential evolving stochastically in time.  In the papers \cite{KS2009,MS2015}, the following stochastic Schr\"odinger equation on $\ell^2(\Z^d)$ was considered:
\begin{equation}
\im \partial_{t}\psi_{t}(x)=H_0\psi_{t}(x) + {\lambda} V(x,t)\psi_{t}(x),\label{eq:pregenSE}
\end{equation}
with $H_0$ a (non-random) translation invariant Schr\"odinger operator, $ {\lambda}\ge0$ a real coupling constant, and $V(x,t)$ a zero-mean random potential with time dependent stochastic fluctuations.  These models had been considered previously by {Tcheremchantsev} \cite{Tcheremchantsev:1997kl,Tcheremchantsev:1998qe},  who obtained diffusive bounds for position moments up to logarithmic corrections. In \cite{KS2009,MS2015}, diffusive scaling for all moments (without logarithms) was proved, under suitable hypotheses on $H_0$ and $V$.  Furthermore, it was observed that at weak disorder $\lambda \rightarrow 0$, the corresponding diffusion constant $D$ has the asymptotic form 
\begin{equation}
D \ \sim \ \frac{C}{\lambda^2}.
\label{eq:DAF}
\end{equation}
The divergence of $D$ as $\lambda \rightarrow 0$ seen in eq.\ \eqref{eq:DAF} is to be expected, since the translation invariant Schr\"odinger operator $H_0$ on its own leads to ballistic transport.
In ref.\ \cite{Schenker:2015}, the first author considered the more subtle situation in which the environment is a superposition of two parts: 
\begin{equation}
\im \partial_{t}\psi_{t}(x)=H_0\psi_{t}(x)+ u (x)\psi_{t}(x)+ {\lambda} V(x,t)\psi_{t}(x),\label{eq:genSE}
\end{equation}
where $u$ is a static random potential that, at $\lambda=0$, gives rise to Anderson localization (absence of transport).  In \cite{Schenker:2015}, it was observed that the diffusion constant in this case has the asymptotic form
\begin{equation}
D \ \sim \ C \lambda^2 .
\label{eq:DALAF}
\end{equation}

Taken together, the results in \cite{KS2009,MS2015,Schenker:2015} suggest that solutions to \eqref{eq:genSE} with a general potential $u$ should satisfy diffusion with a diffusion constant whose asymptotic behavior in the small $\lambda$ limit is governed by the dynamics of the static Schr\"odinger operator $H_0 +  u$. In this paper, we study this idea in the context of models of the form of eq.\ \eqref{eq:genSE} but with \emph{periodic} $u$ that leads to ballistic transport.  We will obtain diffusive propagation for the evolution, and more generally, a central limit theorem for the square amplitude. Furthermore, we prove that in this case the asymptotic relation \eqref{eq:DAF} holds.

We consider below solutions to eq.\ \eqref{eq:genSE} with $\{u(x)\}_{x\in\Z^d}$ a real valued $\vec p$-periodic potential. Recall that  given $\vec p=\{p_j\}_{j=1}^d\in \Z_{>0}^d$, a function $u:\Z^d\mapsto \R$ is called $\vec p$-periodic if 
\begin{align}\label{eq:U-periodic}
u(x+p_je_j)=u(x)
\end{align}
for all $1\le j\le d$ and $x\in\Z^d$, where $e_j$ denotes the standard basis of $\Z^d$. Without loss of generality, we assume that $p_j\ge2$ for some $j$. Otherwise, $u$ is constant and the problem reduces to that studied in \cite{KS2009}. {Throughout this paper, we denote by $U$ the multiplication operator, $(U\psi)(x)=u(x)\psi(x)$ for $\psi(x)\in\ell^2(\Z^d)$.}

The analysis below is applicable to a broad class of operators $H_0$ and $V(x,t)$. To avoid technicalities in this introduction, let us state the main results in terms of hopping $H_0$ given by the standard discrete Laplacian on $\Z^d$ and potential $V(x,t)$ given by the following so-called Markovian ``flip process,'' which is a non-trivial, and somewhat typical, example of a potential satisfying the general requirements. In general, the random potential is given by $V(x,t)=v_x(\omega(t))$, where $\omega(t)$ is an evolving point in an auxiliary state space $\Omega$. For the flip process, we take the state space $\Omega=\{-1,1\}^{\Z^{d}}$, and $v_x(\omega) = \omega_x $, the $x^{\mathrm{th}}$ coordinate of $\omega$.  Thus the potential $V(x,t)= v_x(\omega(t))$ takes only the values $\pm1$. Now suppose the process $\omega(t)$ is obtained by putting independent, identical Poisson processes at each site $x$, and allowing each coordinate $\omega_x$ to flip sign at the times $t_{1}(x) \le t_{2}(x) \le \cdot$ of the Poisson process. Now the general equation \eqref{eq:genSE} becomes:
\begin{equation}\label{eq:SE}
 \im \partial_{t} \psi_{t}(x) \ = \ \sum\limits_{|y-x|=1} \psi_{t}(y) +  u  ({x}) \psi_{t}(x)+ {\lambda} v_x(\omega(t)) \psi_{t}(x).
\end{equation} 

A sign of diffusive propagation is the existence of a \emph{diffusion constant} for eq.\ \eqref{eq:SE} 
	\begin{equation}\label{eq:DC}
	D:= \ \lim_{t \rightarrow \infty} \frac{1}{t} \sum_{x} |x|^2 \E({\abs{\psi_t(x)}^2}),
	\end{equation}
	characterized by the relationship $x\sim \sqrt{t}$ in the mean amplitude of evolving wave packets. Here, and throughout this introduction, $\E({\cdot})$ denotes averaging with respect to the Poisson fliping times $t_1(x)\le t_2(x)\le \cdots$ and the initial values $\{\omega_x\}_{x\in \Z^d}$, taken independent and uniform in $\{-1,1\}$.

	We will show below that the limit in eq.\ \eqref{eq:DC} exists for any $\vec p$-periodic potential $u$ and ${\lambda}> 0$, and furthermore  $D>0$.  To give an unambiguous definition, one may take the initial value $\psi_0(x) = {\delta_{\vec 0}}(x)$.  However, as we will show, the limit remains the same for any other choice of (normalized) $\psi_0$ with $\sum_x |x|^2 \abs{\psi_0(x)}^2 < \infty$.  
	
	We refer to the existence of a finite, positive diffusion constant as in eq.\ \eqref{eq:DC} as \emph{diffusive scaling}.  
More generally, we have the following

\begin{thm}[Central limit theorem]\label{thm:CLT} 
		For any periodic potential $u$ and ${\lambda}>0$, there is a positive definite $d\times d$ matrix $\vec{D}=\vec{D}(\lambda,u)$ such that for any bounded continuous function $f:\R^d \rightarrow \R$ and any normalized $\psi_0\in \ell^2(\Z^d)$ we have
	\begin{equation}\label{eq:CLT}
	\lim_{t\rightarrow \infty} \sum_{x\in \Z^d} f\left ( \frac{x}{\sqrt{t}}  \right ) \Ev{\abs{\psi_t(x)}^2} \ = \ \int_{\R^d} f(\vec{r}) \left ( \frac{1}{ 2\pi   }\right )^{\frac{d}{2}} \e^{-\frac{1}{2 }\ipc{\vec{r}}{\vec{D}^{-1} \vec{r}}} \di \vec{r} ,
	\end{equation}
	where $\psi_t(x)$ is the solution to eq. \eqref{eq:SE} with initial value $\psi_0$.
	If furthermore $\sum_{x} (1+|x|^2) \abs{\psi_0(x)}^2 < \infty$, then diffusive scaling  eq.\ \eqref{eq:DC} holds with the diffusion constant 
	\begin{align}\label{eq:DS}
	D(\lambda) \ = \ \lim_{t\rightarrow \infty} \frac{1}{t}\sum_{x\in \Z^d} \abs{x}^2 \Ev{\abs{\psi_t(x)}^2} \ = \ \tr \vec{D}(\lambda).
	\end{align}
	Moreover, eq.\ \eqref{eq:CLT} extends to quadratically bounded continuous $f$ with $\sup_x (1+|x|^2)^{-1} \abs{f(x)} < \infty$. 
\end{thm}

It is well known that if $\lambda=0$ in \eqref{eq:SE}, then the free periodic Schr\"odinger equation has Bloch-wave solutions and exhibits ballistic motion
by the Floquet theory\footnote{{
Let $J$ be a periodic block Jacobi matrix on $\oplus_{j=1}^m\ell^2(\Z)$, which includes $\Delta+U$ on $\ell^2(\Z)$ as a special case. Let $X$ be the position operator and let  $X(t) =e^{\im t (\Delta+U)}Xe^{-\im t (\Delta+U)}$ be its	Heisenberg time evolution. \emph{Strong ballistic motion} was obtained  for $J$ in \cite{DLY}.  That is, there is a bounded self-adjoint operator $Q$ with $\ker (Q)=\{0\}$ such that for any 
$\psi$ with $X\psi\in\ell^2(\Z)$,$$\lim_{t\to\infty}\frac{1}{t}X(t)\psi=Q\psi. $$ 
}}, see \cite{AK98,DLY}:
		\begin{align}\label{eq:balli-freeSE}
	\lim_{t\to\infty}\frac{1}{t^2}\sum_{x\in \Z^d} \abs{x}^2 { \abs{\ipc{\delta_x}{\e^{-\im t (\Delta+U)}\delta_0}}^2}\in (0,\infty).
	\end{align} 
If we extend the definition of $D(\lambda)$ in \eqref{eq:DS} to  $\lambda=0$, then $D(0)=\infty$. We are primarily interested here in the regime $\lambda\sim0$, although we will demonstrate diffusion for all $\lambda>0$. However, for small $\lambda$ the diffusion constant will be large and have the following asymptotic behavior as $\lambda \to 0$: 
\begin{thm}\label{thm:asym}Under the hypotheses of Thm.\ \ref{thm:CLT}, there is a  positive definite $d\times d$ matrix $\vec D^0$ such that 
	\begin{align}\label{eq:asym}
	\vec{D}(\lambda)  =  \frac{1}{\lambda^2}\left( \vec{D}^0 + o(1)\right)\ 
	\ {\rm and}\ \ 
	D(\lambda)=\tr\vec{D}(\lambda) \ = \ \frac{1}{\lambda^2} \left(\tr\vec{D}^0 +o(1)\right)
	\quad \text{as } \lambda \rightarrow 0.
	\end{align}
\end{thm}

The conclusions of Theorems \ref{thm:CLT} and \ref{thm:asym} are true for eq.~\eqref{eq:genSE} under much more general assumptions on the hopping $H_0$ and the time dependent stochastic potential $V(x,t)$. We will state the general assumptions and  results in Section \ref{sec:gen}.

The rest of the paper is organized as follows: In Sec. \ref{sec:gen}, a more general class of operators is introduced and the main result Theorem \ref{thm:genCLT}, which generalizes Theorems \ref{thm:CLT} and \ref{thm:asym}, is formulated. In Sec. \ref{sec:aug-space}  the basic analytic tools of ``augmented space analysis,'' developed previously in \cite{KS2009,Schenker:2015}, are reviewed.  In Sec. \ref{sec:spe-analysis}, we present the heart of our argument, a block decomposition to study the spectral gap of the induced operator on the augmented space. Sec. \ref{sec:MainResults} is devoted to a proof of the main result.  Certain technical results used below are collected in appendices.

\subsection{History and Conjectures} Before turning to the general framework, let us discuss a history of related work on diffusion, explain the relation of prior works to the present one, and finally describe several conjectures for more general tight binding models.  These conjectures are closely related to, but do
not follow from, the work presented here.

A brief history of related studies is as follows. Ovchinnikov and Erikman obtained diffusion for a Gaussian Markov (``white noise'') potential \cite{Ovchinnikov:1974eu}.  Pillet obtained results on transience of the wave in related models and derived a Feynman-Kac representation \cite{Pillet:1985oq} which we employ here.  Using Pillet's Feynman-Kac formula, Tchermentchansev \cite{Tcheremchantsev:1997kl, Tcheremchantsev:1998qe} showed that position moments exhibit diffusive scaling, up to logarithmic corrections for any bounded potential $u(x)$ in \eqref{eq:genSE}:
\begin{equation}\label{eq:Tchscaling}
t^{\frac{s}{2}} \frac{1}{(\ln t)^{\nu_{-}}}\ \lesssim \ \sum_{x}|x|^{s} \Ev{ \abs{\psi_{t}(x)}^{2}} \ \lesssim \  t^{\frac{s}{2}} ( \ln t )^{\nu_{+}} , \quad t \rightarrow \infty.
\end{equation}
The case $u(x)\equiv0$ (or equivalently, $\vec p=(1,\cdots,1)$) was considered in the previous work \cite{KS2009}, where \eqref{eq:Tchscaling} was shown to hold for $s=2$ with $\nu_{-}=\nu_{+}=0$. Moreover, the central limit theorem \eqref{eq:CLT} and the asymptotic behavior \eqref{eq:asym} were also obtained in \cite{KS2009}. The
proof  in \cite{KS2009} was revisited in \cite{MS2015} to obtain diffusive scaling for all
position moments of the mean wave amplitude. The models studied in \cite{KS2009} are special cases of those considered here.

For a certain class of random potentials $u(x)$, including the case of  an  i.i.d.\ potential,   diffusive scaling  and the central limit theorem were proved in \cite{Schenker:2015}. Moreover, if $H_0+u$ exhibits Anderson localization, then $O(\lambda^2)$ asymptotics \eqref{eq:DALAF} were proved for the diffusion constant.  The arguments in \cite{Schenker:2015} do not require strict independence of the static potential at different sites. However, the \emph{Equivalence of Twisted Shifts} assumption taken in \cite{Schenker:2015} excludes $\vec p$-periodic background potentials, as well as almost-periodic background potentials.  The periodic case falls in an intermediate regime between the period-free case and the i.i.d.\ case. This is a key motivation for us to revisit the proofs in \cite{KS2009} and \cite{Schenker:2015} and develop the current approach to the $\vec p$-periodic case, for both diffusive scaling and limiting behavior.

In \cite{FS2015}, Fr\"ohlich and the first author used the techniques of \cite{Schenker:2015} to study diffusion for a lattice particle governed by a Lindblad equation describing jumps in momentum driven by interaction with a heat bath.  In some sense, this is the quantum analogue of the classical dynamics of a disordered oscillator systems perturbed by noise in the form of a momentum jump process, considered in \cite{BO2011,BH2012}  and reviewed in \cite{BHLLO2014}.  In those works, heat transport is considered in the limit of weak noise in a regime for which transport is known to vanish for the disordered oscillator system without noise.  A key feature of the noise in \cite{BO2011,BH2012} is that energy is conserved in the system with noise; this is necessary so that one can speak about heat flux. 
By contrast, in the present work, and in \cite{KS2009,MS2015,Schenker:2015,FS2015}, energy conservation is broken by the noise.  
Indeed the \emph{only} conserved quantity for the evolution we consider is quantum probability; and it is this quantity which is subject to diffusive transport.

That \emph{diffusive} transport emerges from \eqref{eq:genSE} very much depends on the fact that it is a lattice, or \emph{tight-binding}, equation.  A time-dependent potential  coupled with the unbounded kinetic energy present in continuum models can lead to  stochastic acceleration resulting in super-diffusive, or even super-ballistic, transport. Stochastic acceleration has been well studied in the context of classical systems, see for example \cite{aguerDebievreLafitteParris,rosenbluth,soretDebievre1}.  For quantum systems in the continuum, transport has been studied in the context of Gaussian white-noise potentials \cite{fischer-leschke-muller,fischer-leschke-muller2,jayannavar-kumar1,HKOS2019}, for which the super-ballistic transport $\langle x^2 \rangle \sim t^3$ has been proved.

There are also parallel works on diffusion for the continuum Schr\"odinger equation with Markovian forcing and periodic boundary conditions in space, e.g., \cite{EKS03}. One physical interpretation of this continuous model is as a rigid rotator coupled to a classical heat bath. In \cite{EKS03}, the $H^s$ norm of the wave function is shown to behave as $t^{s/4}$. It is interesting to point out that, as in the present work, the existence of a spectral gap for the Markov generator is essential both for their analysis and the results. In many models with Markovian forcing, the potential $V(x,t)$ is quite rough. However, Bourgain  studied the case where  $V(x,t)$ is analytic/smooth in $x$ and quasi-periodic/smooth in $t$. In \cite{Bo99-1}, he showed that energy may grow logarithmically. 
We refer readers to, e.g., \cite{EK09,N09,W08}, for more work on Sobolev norm growth and controllability
of Schr\"odinger equations with time-dependent potentials.

The proof we present here is a generalization  of that in \cite{KS2009}.  Some of the arguments are essentially standard fare and parallel the work of \cite{KS2009} closely.  However, there are three places in the proof where some substantially new arguments were needed.  First, the Fourier analysis (see Sec. \ref{sec:VFourier}) in our work is more subtle and requires careful consideration due to the periodic potential.  The extension developed here is of independent interest and may  benefit the future study of the limit-periodic and quasi-periodic cases. Secondly, the spectral gap Lemma \ref{lem:L0gap} and the proof of the main results in Sec. \ref{sec:MainResults} are technically more involved in the current work. The interaction between the periodic part and the hopping terms complicates the block decomposition on the augmented space.  Finally, in the present proof, the analysis of the asymptotic behavior of the diffusion constant is quite a bit more involved. In \cite{KS2009}, \eqref{eq:DAF} essentially follows from a formula derived for the diffusion constant in the midst of the proof of diffusion.  Unfortunately, Theorem \ref{thm:asym} in the $\vec p$-period case does not have such a simple proof  and is obtained by a new approach. The proof is based on an interesting observation linking the ballistic motion of the unperturbed part to  the diffusive scaling.  This observation is part of the motivation behind our conjecture below on the more general situations, linking the transport exponent to the limiting behavior of the diffusion constant.

In light of the present work, it is natural to ask what can be said about eq.\ \eqref{eq:genSE} with $u$ a general ergodic/deterministic potential. In particular, 
\begin{enumerate}
	\item Under which hypotheses on $u$ do we  have diffusive propagation over long time scales? 
	\item When diffusion holds, what is the limiting behavior of the diffusion constant with respect to the disorder coupling constant?
\end{enumerate}
Based on the  limiting behavior of the diffusion constant in the periodic case and in the i.i.d case, it is natural to make the following 
\begin{conj}\label{conj:gen}
	For any bounded potential $u(x)$ on $\Z^d$ and any $\lambda>0$, there exist positive, and finite, upper and lower diffusion constants, $\underline{D}(\lambda),\overline{D}(\lambda)\in(0,\infty)$ such that 
	\begin{align}\label{eq:conj}
\underline{D}(\lambda):= \liminf_{t\rightarrow \infty} \frac{1}{t}\sum_{x\in \Z^d} \abs{x}^2 \Ev{\abs{\psi_t(x)}^2}\le \limsup_{t\rightarrow \infty} \frac{1}{t}\sum_{x\in \Z^d} \abs{x}^2 \Ev{\abs{\psi_t(x)}^2}=:\overline{D}(\lambda).
	\end{align} 
Suppose $\Delta+U$ exhibits ballistic motion, then $\underline{D}(\lambda),\overline{D}(\lambda)\sim O({\lambda^{-2}})$ for $\lambda\sim0$. Suppose $\Delta+U$ exhibits dynamical localization, then $\underline{D}(\lambda),\overline{D}(\lambda)\sim O({\lambda^2})$ for $\lambda\sim0$.
\end{conj}
\begin{remark}
1) Similar conjectures can be made for the general equations which will be introduced in Section \ref{sec:gen}. 2) More generally, if the unperturbed equation has transport exponent $\rho\in[0,2]$, then we expect $\underline{D}(\lambda),\overline{D}(\lambda)\sim O({\lambda^{2-2\rho}})$.  3) Here, we also want to bring reader's attention to the recent work \cite{KLRS17}, though not directly relevant to our current paper, on the ballistic transport for the Schr\"odinger operator with limit-periodic or quasi-periodic potential in dimension two. 
\end{remark}

If the unperturbed part is given by the almost Mathieu operators with parameters $g\in\R,\theta,\alpha\in[0,1]$, we have the following  AMO-Markovian equation on $\ell^2(\Z)$:
\begin{align}\label{eq:AMO}
 \im \partial_{t} \psi_{t}(x) \ = \psi_t(x+1)+\psi_t(x-1)+2g\cos2\pi(\theta+x\alpha) \psi_{t}(x)+ {\lambda} v(\omega_{x}(t)) \psi_{t}(x).
\end{align}
\begin{conj}\label{conj:AMO}
For almost every $\theta,\alpha\in [0,1]$, the AMO-Markovian equation has a diffusion constant $D(g,\lambda)\in(0,\infty)$ which is a smooth function for all $(g,\lambda)\in \R\times \R^+$. Moreover, $D(g,\lambda)\sim O({\lambda^2})$ for all $|g|>1$ and $D(g,\lambda)\sim O({\lambda^{-2}})$ for all $|g|<1$. 
\end{conj}

%
%
\section{General assumptions and the main result}\label{sec:gen}
We study a more general class of equations with hopping terms other than nearest neighbor and a perturbing potential $V$ that is not necessarily the ``flip process." More precisely, we shall consider equation \eqref{eq:genSE} in the form
\begin{align}\label{eq:genE}
\im \partial_{t}\psi_{t}(x)=H_0\psi_{t}(x)+ u(x)\psi_{t}(x)+ {\lambda} V_\omega(x)\psi_{t}(x)
\end{align}
Here $u$ is the real-valued, $\vec{p}$-periodic potential as in \eqref{eq:U-periodic} for some $\vec{p} \in \Z^d_{>0}$; $H_0$ is a self-adjoint, short-ranged, translation invariant hopping operator with non-zero hopping along a set of vectors that generate $\Z^d$; $V_{\omega(t)}$ is time-dependent random potential that fluctuates according to a stationary Markov process $\omega(t)$; and $\lambda \geq 0$ is a coupling constant used to set the strength of the disorder. These assumptions will be made precise below. Some assumptions are similar to those in \cite{KS2009} and \cite{Schenker:2015}. They are repeated here for convenience. In particular, our assumptions on the probability space and Markov dynamics remain largely unchanged.

\subsection{Assumptions}\label{sec:ass}
\begin{ass}[Probability space]
	Throughout, let $(\Omega, \mu)$ be a probability space, on which the additive group $\Z^d$ acts through a collection of $\mu$-measure preserving maps.  That is, for each $x \in \Z^d$ there is a \mbox{$\mu$-measure} preserving map, $\tau_{x} : \Omega \to \Omega$, where $\tau_0$ is the identity map and $\tau_x \circ \tau_y = \tau_{x+y}$ for each $x,y \in \Z^d$.  We refer to the maps $\tau_x$, $x\in \Z^d$ as ``disorder translations.''
\end{ass}

\begin{ass}[Markov dynamics]
	The space $\Omega$ is a compact Hausdorff space, $\mu$ is a Borel measure and for each $\alpha\in \Omega$ there is a probability measure  $\mathbb{P}_{\alpha}$  on the $\sigma$-algebra generated by Borel-cylinder subsets of the path space $\mc{P}(\Omega)=\Omega^{[0,\infty)}$.  Furthermore, the collection of these measures has the following properties
	\begin{enumerate}
		\item \emph{Right continuity of paths}: For each $\alpha\in \Omega$, with $\mathbb{P}_{\alpha}$ probability one, every path $t\mapsto \omega(t)$ is right continuous and has initial value $\omega(0)=\alpha$.
		\item \emph{Shift invariance in distribution}: For each $\alpha\in \Omega$ and $x\in \Z^d$, $\bb{P}_{\tau_x \alpha}  =  \bb{P}_{\alpha} \circ \mc{S}_x^{-1}$,
		where $\mc{S}_x(\{\omega(t)\}_{t\ge 0}) = \{ \tau_x \omega(t)\}_{t\ge 0}$ is the shift $\tau_x$ lifted to path space $\mc{P}(\Omega)$.
		\item \emph{Stationary Markov property}:  There is a filtration 
		$\{ \mc{F}_t \}_{t\ge 0}$ on the Borel $\sigma$-algebra of $\mc{P}(\Omega)$ such that $\omega(t)$ is $\mc{F}_t$ measurable and    
		$$ \bb{P}_{\alpha}\left ( \left \{ \omega(t+s) \right \}_{t\ge 0} \in \mc{E} \middle | \mc{F}_s \right )  =    \bb{P}_{\omega(s)}(\mc{E})  $$
		for any measurable $\mc{E} \subset \mc{P}(\Omega)$ and any $s>0$.
		\item \emph{Invariance of $\mu $}: For any Borel measurable $E\subset \Omega$ and each $t>0$,
		$$\int_\Omega \bb{P}_{\alpha} (\omega(t)\in E) \;\mu (\di \alpha) \ = \ \mu (E).$$
	\end{enumerate}
\end{ass} 

We use $\mathbb{E}_\alpha(\cdot)$ to denote averaging with respect to $\bb{P}_\alpha$ and $\Ev{\cdot}$ to denote to the combined average $\int_\Omega \mathbb{E}_\alpha(\cdot) \; \mu(\di \alpha)$ over the Markov paths and the initial value of the process.
Invariance of $\mu $ under the dynamics is equivalent to the identity  $\Eva{f(\omega(t))}  \ = \ \Eva{f(\omega(0))} $ for $f\in L^1(\Omega)$.  An important tool for studying Markov processes is conditioning on the value of a process at a given time.  The proper definition can be found in, e.g. \cite{Schenker:2015}. 
Conditioning on the value of the processes at $t=0$ determines the initial value:
$\Ev{\cdot  | \omega(0)=\alpha }   =   \bb{E}_\alpha (\cdot).$
To the process $\{ \omega(t)\}_{t\ge 0}$, there is associated a  Markov semigroup, obtained by averaging over the initial value conditioned on the value of the process at later times: 
$$ S_{t}f(\alpha) \ := \ \bb{E}   \left (  f(\omega(0))  | \omega(t)=\alpha \right ). $$
As is well known, $S_{t}$ is a strongly continuous contraction semi-group on $L^p( \Omega)$ for $1\le p < \infty$.
The semigroup $S_{t}$ has a generator
\begin{equation}\label{eq:Bdefn}
B f \ := \ \lim_{t \downarrow 0} \frac{1}{t} \left (f - S_{t} f \right ),
\end{equation}
defined on the domain $\mathcal{D}(B)$ where the right hand side exists in the $L^2$-norm. By the Lumer-Phillips theorem, $B$ is a \emph{maximally accretive operator}.  Note that $S_t\id = \id$ by definition, where $\id (\alpha)=1$ for all $\alpha \in \Omega  $. The invariance of $\mu$ under the process $\{\omega(t)\}_{t\ge 0}$ implies further that $S_{t}^\dagger \id  =\id $. It follows that 
$$L^2_0( \Omega ) \ := \ \setb{f\in L^2(\Omega)}{\int_\Omega  f(\alpha) \mu    (\di \alpha) \ = \ 0}$$
is invariant under the semi-group $S_t$ and its adjoint $S_t^\dagger.$ We assume that $B$ is sectorial and strictly dissipative on $L_0^2(\Omega)$. 
\begin{ass}[\emph{Sectoriality of $B$}]\label{ass:sectoriality}
	There are $b,\gamma \ge 0$  such that
	\begin{equation}\label{eq:sector}
	\abs{\Im \ipc{f}{Bf}} \ \le \ \gamma \Re \ipc{f}{B f} + b \norm{f}^2 
	\end{equation}
	for all $f\in \mathcal{D}(B)$.  Here $\ipc{f}{g} = \int_{ } \overline{f} g \di \mu $ denotes the inner product on $L^2( \Omega)$.
\end{ass}


\begin{ass}[\emph{Gap condition for $B$}]\label{ass:gap}  There is $T   >0$ such that 
	\begin{equation}\label{eq:gap} \Re \ipc{f}{B f} \ \ge \  \frac{1}{T}  \norm{f-\int_{\Omega } f \di \mu  }_{L^2( \Omega)}^2
	\end{equation}
	for all $f\in \mathcal{D}(B)$. 
\end{ass}
\begin{remark} 1) The \emph{resolvent} of the semigroup $\e^{-t B}$ is the operator valued analytic function
	$ R(z) \ := \ (B - z)^{-1} \ = \ \int_0^\infty \e^{tz} \e^{-tB} \di t,$
	which is defined and satisfies $\norm{R(z)} \le \frac{1}{\abs{\Re z}}$  when $\Re z <0$. Sectoriality is equivalent to the existence of a analytic continuation of $R(z)$ to $z\in\C \setminus K_{b,\gamma}$ with the bound $ 	\norm{R(z)}  \ \le \ {\dist^{-1}(z,K_{b,\gamma})}$
	where $K_{b,\gamma}$ is the sector $\{\Re z \ge   0\} \cap \{\abs{\Im z } \le b + \gamma \abs{\Re z}\}$ (see \cite[Theorem V.3.2]{Kato:1995sf}). In particular Assumption \ref{ass:sectoriality} holds (with $b=0$ and $\gamma=0$) if the Markov dynamics is reversible, in which case $B$ is self-adjoint.
2) The gap assumption implies that the restriction of $B$ to $L^2_0(\Omega)$ is strictly accretive, and thus that
$\norm{\left. S_t\right|_{L^2_0(\Omega)}}\le \ \e^{-\frac{t}{T}}$.
\end{remark}


\begin{ass}[Translation covariance, boundedness and non-degeneracy of the potential]
	The potentials  $V_{\omega}(x)$ appearing in the Schr\"odinger equation \eqref{eq:genSE} are given by 
	$V_{\omega}(x) \ = \  v(\tau_x\omega)$
	where  $v\in L^\infty(\Omega)$. We assume that $\|v\|_{\infty}=1$,  $\int_{\Omega} v(\omega) \mu  (\di \omega) = 0$, and $v$ is non-degenerate in the sense that there is $\chi >0$  such that
	\begin{equation}\label{eq:nondegV}
	\norm{B^{-1} (v(\tau_x \cdot)- v(\tau_y \cdot))}_{L^{2}(\Omega)} \ge \chi
	\end{equation}
	for all $x, y \in \Z^{d}$, $x \neq y$.
\end{ass}
\begin{remark}
Since the Markov process is translation invariant, $B$ commutes with the translations $T_{x}f(\alpha) = f(\tau_{x}\alpha)$ of $L^{2}(\Omega)$.  Thus \eqref{eq:nondegV} is equivalent to
\begin{equation} \label{eq:BinverseV}
\norm{B^{-1}(v(\tau_x \cdot)- v( \cdot))}_{L^{2}(\Omega)} \ge \chi.
\end{equation}
for all $x \in \Z^{d}$, $x \neq \vec 0$. The non-degeneracy  essentially amounts
to requiring that $B^{-1}(v\tau_x)$ are uniformly non-parallel to $B^{-1}(v)$ for $x\neq0$. In particular, the condition is trivially satisfied 
if for example if the processes $v(\tau_x\omega(t))$ and
$v(\omega(t))$ are independent for $x\neq0$, as in the ``flip process''.
\end{remark}

\begin{ass}[Translation invariance and non-degeneracy of the hopping terms]
	The hopping operator, $H_0$, on $\ell^2(\Z^d)$ is defined by 
	\begin{equation}	
	H_0 \psi(x) = \sum_{\xi \neq x} h(x-\xi)\psi(\xi).
	\end{equation} 
	Additionally, the hopping kernel $h : \Z^d \setminus \{\vec 0\} \to \C$ is 
	\begin{enumerate}
		
		\item Self-adjoint:
		\[
		h(-\xi)  = \overline{h(\xi)};
		\]
		\item Short range:
		\begin{align} \label{eq:h-range}
		\sum_{\xi \in \Z^d \setminus \{\vec 0\}} |\xi|^2|h(\xi)| < \infty;
		\end{align}		
		\item Non-degenerate:
		\begin{align}\label{eq:supp-h-nonden}
		{\rm span}_{\Z}\left({\rm supp}h\right)=\Z^d,
		\end{align}
		where $ {\rm supp}h=\left\{\xi\in\Z^d:\  h(\xi)\neq 0 \right\}
		$.
		\end{enumerate}
\end{ass}
	\begin{remark}
1)	It follows from (1) and (2) that $\wh{h}(\vec{k}) = \sum_{x} \e^{-\im \vec{k} \cdot x} h(x)$ is a real-valued $C^{2}$ function on the torus $[0,2\pi)^d$. In particular, $H_0$ is a bounded self-adjoint operator with $\norm{H_0}_{\ell^{2}(\Z^{d})\to  \ell^{2}(\Z^{d})} = \max_{\vec{k}} |\widehat h(\vec{k})|$ and
	\begin{align}
	\|{\wh{h}}\|_{\infty},
	\|{\wh{h}'}\|_{\infty},
	\|{\wh{h}''}\|_{\infty}\le \sum_{\xi \in \Z^d \setminus \{\vec 0\}} (1+|\xi|^2)|h(\xi)| < \infty. \label{eq:h-norm}
	\end{align}
	2)	It is natural to assume that ${\rm supp}\,h$ can generate the entire $\Z^d$ lattice, otherwise the system can always be reduced a direct sum of systems over several sub-lattices. 
\end{remark}

Below we will need the following simple consequence of the non-degeneracy of $h$:
\begin{prop}For each non-zero $\vec{k} \in \R^d$, 
	\begin{align}\label{eq:h-nonden}
	\sum_{\xi \in \Z^d} |\vec{k}\cdot \xi|^2 |h(\xi)|^2 > 0.
	\end{align}
\end{prop}
\begin{proof}
	Suppose on the contrary that $\sum_{\xi \in \Z^d} |\vec{k}\cdot \xi|^2 |h(\xi)|^2=0$ for some $\vec{k}\neq \vec{0}$.  It follows that $\vec{k}\cdot \xi=0$ for all $\xi \in {\rm supp} \, h$, violating the non-degeneracy of $h$.
\end{proof}

\subsection{General result}
The main result is the following
\begin{thm}[Central limit theorem]\label{thm:genCLT} 
	For any periodic potential $u$ and ${\lambda}>0$, there is a positive definite $d\times d$ matrix $\vec{D}=\vec{D}(\lambda,u)$  such that for any bounded continuous function $f:\R^d \rightarrow \R$ and any normalized $\psi_0\in \ell^2(\Z^d)$ we have
	\begin{equation}\label{eq:genCLT}
	\lim_{t\rightarrow \infty} \sum_{x\in \Z^d} f\left ( \frac{x}{\sqrt{t}}  \right ) \Ev{\abs{\psi_t(x)}^2} \ = \ \int_{\R^d} f(\vec{r}) \left ( \frac{1}{ 2\pi   }\right )^{\frac{d}{2}} \e^{-\frac{1}{2 }\ipc{\vec{r}}{\vec{D}^{-1} \vec{r}}} \di \vec{r},
	\end{equation}
		where $\psi_t(x)$ is the solution to eq. \eqref{eq:genE}.
	If furthermore $\sum_{x} (1+|x|^2) \abs{\psi_0(x)}^2 < \infty$, then diffusive scaling  eq.\ \eqref{eq:DC} holds with the diffusion constant 
	\begin{align}\label{eq:genDS}
	D(\lambda) \ = \ \lim_{t\rightarrow \infty} \frac{1}{t}\sum_{x\in \Z^d} \abs{x}^2 \Ev{\abs{\psi_t(x)}^2} \ = \ \tr \vec{D}(\lambda).
	\end{align}
	Moreover, eq.\ \eqref{eq:genCLT} extends to quadratically bounded continuous $f$ with $\sup_x (1+|x|^2)^{-1} \abs{f(x)} < \infty$.

	Assume further that 
	\begin{align}\label{eq:balli-gen}
\lim_{T\to\infty}\,\frac{2}{T^3}\int_0^\infty \e^{-\frac{2t}{T}}\, \sum_{x\in \Z^d} x_j^2 { \abs{\ipc{\delta_x}{\e^{-\im t (H_0+U)}\delta_0}}^2}\, {\rm d}t>0,\ \ j=1\cdots,d,
	\end{align}
	then there is a  positive definite $d\times d$ matrix $\vec D^0$ such that 
	\begin{align}\label{eq:genAsym}
\vec{D}(\lambda)  =  \frac{1}{\lambda^2} \left(\vec{D}^0 + o(1)\right)\ 
	\ {\rm and}\ \ 
	D(\lambda)=\tr\vec{D}(\lambda) \ = \ \frac{1}{\lambda^2} \left(\tr\vec{D}^0 +o(1)\right)
	\quad \text{as } \lambda \rightarrow 0.
	\end{align}
\end{thm}

\begin{remark}
1) In the case with the short range hoping $H_0$ and periodic $U$, the strong limit of all the $j$-th velocity operators $\lim_tt^{-1}\, X_j(\psi_t)$ always exist, which implies the existence of the limit in \eqref{eq:balli-gen}. We say $H_0+U$ has ballistic motion if the limit in \eqref{eq:balli-gen} is positive.  2) $\delta_{\vec 0}$ in  \eqref{eq:balli-gen} can be replaced by any $\psi_0$ with compact support. 3) There always exists a \emph{semi}-positive definite $d\times d$ matrix $\vec D^0$ such that \eqref{eq:genAsym} holds regardless of \eqref{eq:balli-gen}. If \eqref{eq:balli-gen} is true for $j\in S$ with $S\subset \{1,2,\cdots,d\}$, then the restriction of $\vec D^0$  on $S\times S$ is positive definite, and we still have $	D(\lambda)\sim \lambda^{-2}$ since $\tr\vec{D}^0>0$. 
\end{remark}


\section{Augmented space analysis}\label{sec:aug-space}
\subsection{The Markov semigroup on  augmented spaces and the Pillet-Feynman-Kac formula} As in the works \cite{KS2009,Schenker:2015}, our analysis of the Schr\"odinger equation eq.\ \eqref{eq:genE} is based on a formula of Pillet \cite{Pillet:1985oq} for $\E (\rho_t)$, where $\rho_t(x,y)  =  \psi_t(x)\overline{{\psi_t}({y})}$ is the density matrix  corresponding to a solution $\psi_t$ to eq.\ \eqref{eq:genE}. Pillet's formula relates $\E (\rho_t)$ to matrix elements of a contraction semi-group on the ``augmented space'' 
\begin{equation}\label{eq:augmented}
	{ \mathcal  H} \ := \ L^2(\Omega;{ \mathcal  HS}(\Z^d) ) ,
\end{equation}
where  ${ \mathcal  HS}(\Z^d)$ denotes the Hilbert-Schmidt ideal in the bounded operators on $\ell^2(\Z^d)$.

The term ``augmented space'' refers to a space of functions obtained by ``augmenting'' functions defined on $X=\Z^d$ or $X=\Z^d\times\Z^d$ by allowing dependence on the disorder $\omega\in \Omega$. More specifically, it refers spaces of the form
\begin{defn}[Definition 3.1 of \cite{Schenker:2015}]
	Let $(\mc{B}(X), \|\cdot\|_{\mc{B}(X)})$ be a Banach space of functions on $X$ whose norm satisfies
	\begin{enumerate}
		\item If $g \in \mc{B}(X)$ and $0 \leq |f(x)| \leq |g(x)|$ for every $x \in X$, then $f \in \mc{B}(X)$ and $\|f\|_{\mc{B}(X)} \leq \|g\|_{\mc{B}(X)}$. 
		\item For every $x \in X$, the evaluation $x \mapsto f(x)$ is a continuous linear functional on $\mc{B}(X)$. 
	\end{enumerate}
	For $p \geq 1$, the \textbf{augmented space} $\mc{B}^p(X \times \Omega)$ is the set of maps $F: X \times \Omega \to \C$ such that $\|F(x,\cdot)\|_{L^p(\Omega)} \in \mc{B}(X)$. 
\end{defn}
A general theory of such spaces is developed in \cite{Schenker:2015}. In particular, it is shown there that $\mc{B}^p(X\times \Omega)$ is a Banach space under the norm 
\[\label{eq:Bpin}
\|F\|_{\mc{B}^p(X \times \Omega)} :=  \left\|\left( \int_\Omega |F(x, \omega)|^p \, \mu(d\omega) \right)^{\frac{1}{p}} \right\|_{\mc{B}(X)},\]
with
$\norm{F}_{\mc{B}^p(X\times \Omega)} \le \left (\int_\Omega \norm{F(\cdot,\omega)}^p\mu(\di x)\right )^{\frac{1}{p}}$ \cite[Prop. 3.1]{Schenker:2015}.   It follows that $L^p(\Omega;\mc{B})\subset \mc{B}^p(X\times \Omega)$, although in general equality may not hold.  For $\mc{B}(X)= \ell^p(X)$ and $1\le q\le \infty$, we denote $\mc{B}^{q}(X)$ by $\ell^{p;q}(X).$   Then, for $1\le p <\infty$, 
$$\ell^{p;p}(X\times \Omega)= L^p(\Omega;\ell^p(X))= L^p(X\times \Omega),$$
where we take the product measure $\text{Counting Measure}\times \mu$ on $X\times \Omega$ \cite[Prop 3.2]{Schenker:2015}.  In particular, $\ell^{2;2}(X\times \Omega)$ is a Hilbert space with inner product
$$ \langle F , \ G \rangle \ = \ \sum_{x\in X} \int_{\Omega} \overline{F({x, \omega})} G({x, \omega}) \mu({\di \omega}).$$
Another space that will play an important role below is $\ell^{\infty;1}(X\times \Omega)$ which is the space of maps with
$$\norm{F}_{\ell^{\infty;1}} \ := \ \sup_{x\in X} \int_\Omega  \abs{F(x,\omega)} \mu(\di \omega) \ < \ \infty.$$

Returning now to ${ \mathcal  H} =  L^2(\Omega;{ \mathcal  HS}(\Z^d) )$, we note that we may think of an element $F\in { \mathcal  H}$  as a $\bb{C}$-valued map on
\begin{equation}\label{eq:Mdefn}
M \ := \ \Z^d \times \Z^d \times \Omega,
\end{equation} 
via the identification
\begin{equation}\label{eq:Fidentify} F({x, y,\omega}) \ := \ \langle{\delta_x},{F(\omega) \delta_y} \rangle .\end{equation}
It follows from \cite[Prop. 3.2]{Schenker:2015} that
\[\label{eq:calH} { \mathcal  H} \ = \ \ell^{2;2}(\Z^d \times \Z^d \times \Omega) \ = \ L^2(M),
\]  
provided $M$ is given the product measure $m = \left(\mathrm{counting} \ \mathrm{measure}\ \mathrm{on}\ \Z^d\times \Z^d\right)\times \mu $. 

We define operators $\mc{K}$, $\mc{U}$ and $\mc{V}$ that lift the commutators with  $H_0$, $U$ and $V_{\omega}$ to ${ \mathcal  H}$:
\begin{multline}
\mc{K}F(\omega) \  := \ \com{H_0}{F(\omega)} ,\quad \mc{U}F(\omega) \ := \ \com{U}{F(\omega)}, \\ \  \text{and} \quad \mc{V}F( \omega) \ := \ \com{V_{\omega}}{F(\omega)}.\label{eq:KVdefn}
\end{multline}
The following proposition follows immediately from eq.\ \eqref{eq:KVdefn}.
\begin{prop} The operators $\mc{K}$, $\mc{U}$ and $\mc{V}$ are self-adjoint, bounded and are given by the following explicit expressions
	\begin{align}\label{eq:Kform}
	\mc{K} F(x,y,\omega) \ =& \ \sum_{\xi\neq\vec 0} h(\xi)\left[F(x-\xi, y, \omega) - F(x, y- \xi, \omega) \right], \end{align}
	\begin{equation}\label{eq:Uform}
	\mc{U}F(x,y,\omega) \ = \ \left [ u(x)-u(y)\right ] F(x,y,\omega)
	\end{equation}
	and
	\begin{equation}\label{eq:Vform}\mc{V} F(x,y,\omega) \ = \ \left [ v(\tau_x\omega) -v(\tau_y\omega) \right ]F(x,y,\omega), \end{equation}
	for any $F\in L^2(M)$.
\end{prop}

The final ingredient for Pillet's formula is the lift of the Markov generator $B$ to $L^2(M)$. Throughout, we will use $\e^{-tB} $ to denote the Markov semigroup lifted to the augmented space ${ \mathcal  B}^p(X\times\Omega)$, with $B$ the corresponding generator. This semigroup is defined by
\begin{equation}\label{eq:sgext} \e^{-t B} F({x, \alpha}) \ := \ \E_\Omega\left({F(x,\omega(0))}\,|\,{\omega(t)=\alpha}\right) .\end{equation}
In particular, given $\phi \in \mc{B}(X)$ and $f \in L^p(\Omega)$ we have 
\[\e^{-tB}(\phi\otimes f) \ = \ \phi \otimes \e^{-tB}f , \]
where $\phi \otimes f$ denotes the function 
		\[
			(\phi\otimes f )(x,\omega) := \phi(x)f(\omega). 
		\]
\begin{prop}[Prop. 3.3 of \cite{Schenker:2015}]\label{prop:sgext}  The semigroup $\e^{-tB}$ is contractive and positivity preserving on ${ \mathcal  B}^p(X\times\Omega)$  and $B$ is sectorial on $L^2(X\times\Omega)$, with the same constants $b$ and $\gamma$ as appear in Assumption \ref{ass:sectoriality}.
\end{prop}

Pillet's formula expresses the average of the time dependent dynamics \eqref{eq:SE} in terms of the semi-group on $L^2(M)$ generated by $\mc{L} = \im\mc{K} + \im \mc{U} + \im \lambda \mc{V} + B$.
\begin{lem}[Pillet's formula \cite{Pillet:1985oq}]\label{lem:pillet} Let 
	\begin{equation}\label{eq:L} \mc{L}  \ := \ \im \mc{K} + \im \mc{U} + \im {\lambda} \mc{V}    +B
	\end{equation}
	on the domain $\mc{D}(B) \subset L^2(M)$.   Then $\mc{L}$ is maximally accretive and sectorial and if $\rho_t = \psi_t \ipc{\psi_t}{\cdot}$ is the density matrix corresponding to a solution $\psi_t$ to eq.\ \eqref{eq:SE} with $\psi_0 \in\ell^2(\Z^d)$, then 
	\begin{equation}\label{eq:Pillet}
	\Evac{\rho_t}{\omega(t)=\alpha } \ = \ \e^{-t \mc{L}} \left (\rho_0\otimes \mathbbm{1}  \right ),
	\end{equation}
	where $\mathbbm{1}(\omega) = 1$ for all $\omega$. Consequently, we have
	\begin{equation}\label{eq:integratedPillet}
	\Ev{\rho_t} \ = \ \int_{ \Omega} \left [ \e^{-t \mc{L}} \left ( \rho_0\times \mathbbm{1}  \right ) \right ](\omega) \mu({\di \omega}).
	\end{equation}
	Furthermore, for a solution $\psi_t$ to eq.\ \eqref{eq:genE}, we have 
	\begin{align}\label{eq:PFK}
	\Ev{\psi_t(x) \overline{\psi_t(y)}} \ = \ \ipc{ \delta_x\otimes \delta_y\otimes\id}{\e^{-t\mc{L}} \left (   \psi_0 \otimes \overline{\psi_0} \otimes \id \right )}_{L^2(M)} .
	\end{align}
	In particular, we have 
	\begin{equation} 
	\Ev{\rho_{t}(x,x)} = \ip{ \delta_{x}\otimes \delta_{x} \otimes \id,  \e^{-t\mc{L}} \rho_{0} \otimes {\mathbbm 1}}_{{ \mathcal  H}} \label {E(Rho2)}.
	\end{equation}
\end{lem}
\begin{remark}Here and below we will use tensor product notation for elements of $\ell^2(\Z^d \times \Z^d)$,
	$$[\phi \otimes \psi](x,y)  \ = \ \phi(x) \psi(y).$$
	Thus a rank one operator $\psi \ipc{\phi}{\cdot}\in \mc{HS}(\Z^d)$ corresponds to $\psi \otimes \overline{\phi}$.
\end{remark}

For the derivation of this result, we refer the reader to \cite[Lemmas. 3.5 and 3.6]{Schenker:2015}.  In \cite{Schenker:2015}, the term $\mc{U}$ is different, stemming as it does there from the background static random potential.  However, an essentially  identical proof works in the present context.

\subsection{Vector valued Fourier Analysis}\label{sec:VFourier}
For each $\xi \in \Z^d$, we define the (simultaneous position and disorder) shift operator 
\begin{equation}
S_{\xi} \Psi(x,y,\omega) := \Psi(x -\xi, y-\xi, \tau_{\xi}\omega)
\end{equation}
for any function $\Psi$ defined on $\Z^d \times \Z^d \times \Omega$. 
\begin{prop}\label{prop:vanishingcom1} The map $\xi \mapsto S_\xi$ is a unitary representation of the additive group $\Z^d$ on the Hilbert space $\mc{H}$,
	and for every $\xi \in \Z^d$ 
	\[\com{S_\xi}{\mc{K}} \ = \  \com{S_\xi}{\mc{V}} \ = \ \com{S_\xi}{B} \ = \ 0.\]
\end{prop}

The potential term $\mc{U}$ only commutes with a subgroup of translations $S_{\xi}$, corresponding to translation over a period of the potential. For $\xi \in \Z^d$ let
\begin{equation}\label{eq:pxi}
\vec{p}\circ\xi \ := \ (p_1\xi_1,\ldots,p_d\xi_d)
\end{equation}
and 
\begin{equation}
{\vec p}\Z^d=\{\vec{p}\circ\xi:\  \xi\in\Z^d\}.
\end{equation}
Then
\begin{prop}\label{prop:vanishingcom2}For every $\xi \in \Z^d$, $\com{S_{\vec{p}\circ \xi}}{\mc{U}}  =  0.$
\end{prop}

Because of Props.\ \ref{prop:vanishingcom1}, \ref{prop:vanishingcom2}, a suitable Floquet transform will give a fibre decomposition of the various operators $\mc{K}$, $\mc{U}$, $\mc{V}$ and $B$.  Let $\T^d= [0,2\pi)^d$ denote the torus, 
\begin{equation*}
\wh{M} \ := \ \Z^d \times \Omega ,
\end{equation*} 
and let $\Z_{\vec{p}}=\Z_{p_1} \times \cdots \times \Z_{p_d}$ denote the fundamental cell of the periodicity group on $\Z^d$. Note that $\ell^2(\Z_{\vec{p}}) \cong \C^{\otimes \vec{p}} := \C^{p_1} \otimes \cdots \otimes \C^{p_d}$.  Using this identification,  let
$\pi_\sigma : \C^{\otimes \vec{p}} \to \C$ be the coordinate evaluation map associated to a point $\sigma = (\sigma_1, \cdots, \sigma_d) \in \Z_{\vec{p}}$.  For $f,g\in L^2(\wh{M};\C^{\otimes\vec{p}}) $,  we use the natural inner product on $L^2(\wh{M};\C^{\otimes\vec{p}})$ 
\begin{align}\label{def:inner-product-Cp}
\ipc{f}{g}_{L^2(\wh{M};\C^{\otimes\vec{p}})}=\sum_{\sigma\in \Z_{\vec p} }\ipc{\pi_\sigma f}{\pi_\sigma g}_{L^2(\wh{M};\C)}.
\end{align}

Given $\Psi \in L^2(M)$ and ${\vec k}\in \T^d$, the \emph{Floquet transform of $\Psi\in L^2(M)$ at $\vec{k}$} is defined to be a map $\wh{\Psi}_{\vec k}:\wh{M} \rightarrow \C^{\otimes\vec{p}} $ as follows: 
\begin{align}\label{eq:Fourier-vector}
\pi_\sigma\wh{\Psi}_{\vec k}(x,\omega)  :=  &\sum_{\xi \in \Z^d} \e^{-{\im}\, {{\vec k}\cdot (\vec{p}\circ\xi+\sigma)} \, }S_{\vec{p}\circ\xi+\sigma} \Psi(x, 0, \omega)\\ = &\sum_{n\in {\vec p}\Z^d+\sigma} \e^{-{\im}\, {{\vec k}\cdot n} \, }\Psi(x-n, -n, \tau_{n}\omega), \nonumber
\end{align}
for each $\sigma \in \Z_{\vec p}$. Initially we define this Floquet transform on the augmented space 
	\begin{equation}\label{eq:W1M}
	\mc{W}^1(M)\ := \ \setb{F:M\to\C }{\sup_x\sum_y\int|F(x+y,y,\omega)|\mu(\di \omega)<\infty }.
	\end{equation}
	The basic results of Fourier analysis are naturally extended to this Floquet transform. In particular, if $F\in \mc{W}^1(M)$, then $\wh{F}_{\vec{k}}\in \ell^{\infty;1}(\wh{M})$ for each $\vec{k}$ and $\vec{k} \mapsto \wh{F}_{\vec{k}}$ is continuous.  Furthermore, Plancherel's Theorem,
	\[ \norm{F}_{L^2(M)}^2 \ = \ \int_{\mathbb{T}^d} \|\wh{F}_{\vec{k}}\|_{L^2(\wh{M})}^2 \nu(\di \vec{k}),\]
	 holds for $F\in \mc{W}^1(M)\bigcap L^2(M)$, where $\nu$ denotes normalized Lebesgue measure on the torus $\T^d$. Thus, the Floquet transform
	extends naturally to $L^2(M)$. Throughout the rest of the paper, we assume that the Floquet transform is properly defined on $L^2(M)$. For more details of this extension in a similar context, we refer readers to Sec. 3 in \cite{Schenker:2015}.
%
%

One may easily compute 
\begin{align*}
\pi_\sigma\, \wh{(\mc{K}\Psi)}_{\vec k}(x,\omega)
=& \sum_{\xi\neq\vec 0} h(\xi)\left[\pi_\sigma\wh{\Psi}_\vec{k} (x-\xi, \omega) - \e^{-\im \vec{k}\cdot \xi}\pi_{\sigma-\xi}\wh{\Psi}_\vec{k}(x-\xi, \tau_\xi \omega)\right];\\
\pi_\sigma\, \wh{(\mc{U}\Psi)}_{{\vec k}}(x, \omega)
=& \left(u({x-\sigma}) - u({-\sigma})\right)\pi_\sigma\, \wh{\Psi}_{{\vec k}}(x, \omega);\\
\pi_\sigma\, \wh{(\mc{V}\Psi)}_{{\vec k}}(x, \omega)
=& \left(v(\tau_{x}\omega)-v(\omega) \right)\pi_\sigma\, \wh{\Psi}_{{\vec k}}(x, \omega);\\
\pi_\sigma\, \wh{(B\Psi)}_{{\vec k}}(x, \omega) =& B\,  \pi_\sigma\, \wh{\Psi}_{{\vec k}}(x, \omega),
\end{align*}
where on the right hand side, $B$ acts on $\pi_\sigma \wh{\Psi}_{\vec{k}}$ as in eq.\ \eqref{eq:sgext}. With the above computations in mind, let $\wh{\mc{K}}_\vec{k}$, $\wh{\mc{U}}$, and $\wh{\mc{V}}$ denote the following operators on functions  ${\phi}:\wh{M} \rightarrow \C^{\otimes\vec{p}}$:
\begin{equation} \label{eq:Kk}
\pi_\sigma(\wh{\mc{K}}_{\vec k}\, \phi)\,(x,\omega) = \sum_{\xi\neq\vec 0} h(\xi)\left[\pi_\sigma \phi(x-\xi, \omega) - \e^{-\im \vec{k}\cdot \xi}\pi_{\sigma-\xi}\phi (x-\xi, \tau_\xi \omega)\right];
\end{equation}
\begin{align}\label{eq:U}
\pi_\sigma\, (\wc{U}\, \phi)(x, \omega)=\left(u({x-\sigma}) - u({-\sigma})\right)\pi_\sigma\phi(x, \omega); 
\end{align}
and
\begin{align}\label{eq:V}
(\wc{V}\phi)(x, \omega) \left(v(\tau_{x}\omega)-v(\omega) \right)\, \phi(x, \omega).
\end{align}


We now present three Lemmas (Lems. \ref{lem:hatKUV}-\ref{lem:semi-Lk}), which describe the basic properties of the operators $\wc{K}_{\vec k}$, $\wc{U}$, and $\wc{V}$. These results are the adaptation to the present context of Lemmas 3.13-3.15 of \cite{Schenker:2015}, with the main difference being that here we consider the vector valued space $L^2(\wh{M};\C^{\otimes\vec{p}})$ instead of $L^2(\wh{M};\C)$.  We omit the details of the proofs here. 

\begin{lem}\label{lem:hatKUV}   Let $\wh{M}=\Z^d\times\Omega,\wc{K}_{\vec k}$, $\wc{U}$ and $\wc{V}$ be given as above, then  
	\begin{enumerate}
		\item $\wc{K}_{\vec k}$, $\wc{U}$ and $\wc{V}$ are bounded on $\ell^{\infty;1}(\wh{M};\C^{\otimes\vec{p}})$. 
		\item  $\wc{K}_{\vec k}$, $\wc{U}$ and $\wc{V}$ are bounded and self-adjoint on $L^2(\wh{M};\C^{\otimes\vec{p}})$ with the following bounds:
		 \begin{align*}
\norm{\wc{K}_{\vec k}}_{L^2(\wh{M};\C^{\otimes\vec{p}})}\le  2\|\wh h\|_\infty,\ \  \norm{\wc{U}}_{L^2(\wh{M};\C^{\otimes\vec{p}})}\le 2\|u\|_\infty, \ \  
\norm{\wc{V}}_{L^2(\wh{M};\C^{\otimes\vec{p}})}\le 2
\end{align*}			

		\item If $\Psi \in L^2(M;\C)$ and let $\wh{\Psi}_{\vec k}$ be given as in \eqref{eq:Fourier-vector}, then 
		\begin{equation*}
		\wh{(\mc{K} \Psi) }_{\vec k} \ = \ \wc{K}_{\vec k} \wh{\Psi}_{\vec k} , \quad  (\wh{\mc{U} \Psi})_{\vec k} \ = \ \wc{U}\, \wh{\Psi}_{\vec k} \quad \text{and} \quad   (\wh{\mc{V} \Psi})_{\vec k} \ = \ \wc{V}\, \wh{\Psi}_{\vec k}
		\end{equation*}
for $\nu$-almost every ${{\vec k}}\in \T^d$.
	\end{enumerate}
\end{lem}

Because the Markov process has a distribution invariant under the shifts, the Markov semigroup commutes  with Floquet transform: 
\begin{lem}[Lemma 3.14,\cite{Schenker:2015}]\label{lem:semi-B} Let the Markov semigroup $\e^{-tB}$ be defined 
	as in eq.\ \eqref{eq:sgext}. Then, 
	$$\wh{\left [ \e^{-t B} \Psi \right ]}_{{\vec k}} \ = \ \e^{-t B}\wh{\Psi}_{{\vec k}}$$
	for $\Psi\in L^2(M)$ and $\nu$-almost every  ${\vec k}\in \T^d$.
\end{lem}

\begin{lem}\label{lem:K'} Let $\wc{K}_{\vec k}$ be given as in \eqref{eq:Kk} with $h$ that satisfies \eqref{eq:h-range}. Then the map ${{\vec k}} \mapsto \wc{K}_{\vec k}$  is $C^2$ on $\T^d$, considered either as a map into the bounded operators on $\ell^{\infty;1}(\wh{M}; \C^{\otimes \vec p})$ or as a map into the bounded operators on $L^2(\wh{M}; \C^{\otimes \vec p})$. 
	
Moreover, we have the explicit expression for the derivatives for any $\phi(x,\omega)\in L^2(\wh{M};\C^{\otimes \vec p})$ , $\vec k\in \T^d$ and $1\le i,j\le d$: 
	\begin{equation}
	\pi_\sigma\partial_{k_j}\wc{K}_\vec{k} \phi(x,\omega) =\im\, \sum_{\xi\neq\vec 0}\,  \xi_j\, h(\xi)\,\e^{-\im \vec{k}\cdot \xi}\pi_{\sigma-\xi}\phi(x-\xi, \tau_\xi \omega),
	\end{equation}
	\begin{equation}
	\pi_\sigma\partial_{k_i}\partial_{k_j} \wc{K}_\vec{k}\phi(x,\omega) = \sum_{\xi\neq\vec 0}\, \xi_i\, \xi_j\, h(\xi)\, \e^{-\im \vec{k}\cdot \xi} \pi_{\sigma-\xi} \phi(x-\xi, \tau_\xi \omega). 
	\end{equation}
	with bounds 
	\begin{align} \label{eq:norm-K'K''}
	\norm{\partial_{k_j}\wc{K}_\vec{k}}\le \|\wh h'\|_\infty,
	\ \ \norm{ \partial_{k_i}\partial_{k_j}\wc{K}_\vec{k}}\le \|\wh h''\|_\infty,
	\end{align}
where $\|\wh h'\|_\infty, \|\wh h''\|_\infty$ are bounded in \eqref{eq:h-norm}. 

In particular, let $\overrightarrow{1}\in \C^{\otimes\vec{p}}$ be the vector with $\pi_{\sigma}\overrightarrow{1}=1$ for all $\sigma\in \Z_{\vec p} $.  Then   
\begin{align}
\partial_{k_j}\wc{K}_\vec{0}\, \delta_{\vec 0}\otimes \overrightarrow{1}\otimes \id &= \im\, \sum_{\xi\neq\vec 0}  \,\xi_j\, h(\xi)\, \delta_{\xi}\otimes \overrightarrow{1}\otimes \id, \label{eq:K0'-delta0}\\
\partial_{k_i}\partial_{k_j} \wc{K}_\vec{0}\, \delta_{\vec 0}\otimes \overrightarrow{1}\otimes \id &= \sum_{\xi\neq\vec 0}\, \xi_i\, \xi_j\, h(\xi)\, \delta_{\xi}\otimes \overrightarrow{1}\otimes \id.  \label{eq:K0''-delta0}
\end{align}
\end{lem}
\begin{remark}
Throughout the rest of the paper, we will frequently use the notation $\ora{1}_q\in \C^q$ for any $q\in \Z_{>0}$ to indicate the constant vector in $ \C^q$ with all entries $1$ and write $\ora{1}=\ora{1}_{\otimes\vec p}$ for simplicity.
\end{remark}

Putting these results together we obtain
\begin{lem}\label{lem:semi-Lk} For each ${\vec k}\in \T^d$, let
	\begin{equation}\label{def:Lk}
	\wc{L}_{{\vec k}}  :=  \im\, \wc{K}_{{\vec k}} + \im \wc{U} + \im {\lambda} \wc{V}  + B
	\end{equation}
	on the domain $\mc{D}(B) \subset L^2(\wh{M};\C^{\otimes\vec{p}})$.  Then $\wc{L}_{{\vec k}}$ is maximally accretive on $L^2(\wh{M};\C^{\otimes\vec{p}})$. Furthermore 
	\begin{enumerate}
		\item For $t>0$, ${\vec k} \mapsto \e^{-t\wc{L}_{{\vec k}}}$ is 
		\begin{enumerate}
			\item a $C^2$ map from $\T^d$ into the contractions  on $L^2(\wh{M};\C^{\otimes\vec{p}})$; and
			\item a $C^2$ map from $\T^d$ into the bounded operators on $\ell^{\infty;1}(\wh{M};\C^{\otimes\vec{p}})$.
		\end{enumerate}
		
		\item The operators $\{ \wc{L}_{{\vec k}} \}_{{\vec k}\in \T^d}$ are uniformly sectorial; that is for every ${\vec k}\in \T^d$ and every $f\in L^2(\wh{M};\C^{\otimes\vec{p}})$
		\begin{equation}\label{eq:sector-Lk} \abs{\Im \ipc{f}{\wc{L}_{{\vec k}}f}} \ \le \ \gamma \Re \ipc{f}{\wc{L}_{{\vec k}}f} + b' \norm{f}_{L^2}^2
		\end{equation}
		where $\gamma,b$ are given as in \eqref{eq:sector} and $b'= 2b + 2\|{\wh h}\|_{\infty} + 2\norm{u}_{\infty} + 2\lambda .$
		\item If $\Psi\in  \mc{W}^1(M)$, then
		\begin{equation}\label{eq:semigFourier}\e^{-t\wc{L}_{{\vec k}}} \wh{\Psi}_{{\vec k}} \ := \  \wh{\left [ \e^{-t\mc{L}} \Psi\right ]}_{{\vec k}}
		\end{equation}
		for every $\vec{k}\in \mathbb{T}^d$.  For $\Psi\in L^2(M)$, eq.\ \eqref{eq:semigFourier} holds for $\nu$-almost every ${\vec k}$.
	\end{enumerate}
\end{lem}

Combining \eqref{eq:semigFourier} with Pillet's formula (Lemma  \ref{lem:pillet}), we obtain the following Floquet transformed Pillet formula in vector form:
\begin{lem}[Floquet transformed Pillet formula]
Let $\psi_0\in \ell^2(\Z^d)$ and define  $\wh{\rho}_{0;{\vec k}}(x)\in \C^{\otimes\vec{p}}$ for $x\in\Z^d,{\vec k}\in\T^d$ as
\begin{align}\label{eq:rho-hat}
\pi_{\sigma}\wh{\rho}_{0;{\vec k}}(x) \ :=
\sum\limits_{n\in {\vec p}\Z^d+\sigma} \e^{-\im {\vec k \cdot n}} \psi_0(x-n) \overline{\psi_0(-n)} ,\  \sigma\in \Z_{\vec p} .
\end{align}
Then 
\begin{multline}\label{eq:FTPillet}\sum_{y\in\Z^d } \e^{-\im {\vec k}\cdot y } \Ev{\psi_t(x-y) \overline{\psi_t(-y)} } \\ = \ \ipc{ {\delta_{x}}\otimes\ora{1}\otimes \id}{\e^{-t \wc{L}_{{\vec k}}} \left (  \wh{\rho}_{0;{\vec k}}  \otimes \id \right ) }_{L^2(\wh{M};\C^{\otimes\vec{p}})} ,
\end{multline}
where $\psi_t$ is the solution to eq.\ \eqref{eq:genE} with initial condition $\psi_0$. Here $ \e^{-t \wc{L}_{{\vec k}}} ( \wh{\rho}_{0;{\vec k}}\otimes \id ) \in  \ell^{\infty;1}(\wh{M};\C^{\otimes\vec{p}})$ for each ${\vec k}$ and is in $L^2(\wh{M};\C^{\otimes\vec{p}})$ for $\nu$-almost every ${\vec k}$.  

In particular, for every ${\vec k}\in \T^d$, 
\begin{equation}\label{eq:Fourier-PFK}
\sum_{x\in\Z^d } \e^{\im {\vec k} \cdot x } \Ev{|\psi_t(x)|^2 } \ = \ \ipc{ {\delta_{\vec 0}}\otimes\ora{1} \otimes\id}{\e^{-t \wc{L}_{{\vec k}}} \left (  \wh{\rho}_{0;{\vec k}}  \otimes \id \right ) }_{L^2(\wh{M};\C^{\otimes\vec{p}})} .
\end{equation}
\end{lem}

\begin{proof}
	Let $\Psi(x,y,\omega)=\left(\e^{-t\mc{L}}(\rho_0\otimes\id)\right)(x,y,\omega)
	=\ipc{\delta_x\otimes\delta_y}{\e^{-t\mc{L}}(\rho_0\otimes\id)}_{L^2(\Z^d\times\Z^d)}$. Pillet's formula \eqref{eq:PFK} can be rewritten as 
	\begin{align*}
	\Ev{\Psi(x,y,\cdot)}=&\int_\Omega\left(\e^{-t\mc{L}}(\rho_0\otimes\id)\right)(x,y,\omega)\,\mu(\di\omega)\\
	=&\ipc{\delta_x\otimes\delta_y\otimes\id}{\e^{-t\mc{L}}(\rho_0\otimes\id)}_{L^2(\Z^d\times\Z^d\times\Omega)}
	=\Ev{\psi_t(x) \overline{\psi_t(y)} }. 
	\end{align*}
	We note that $\rho_0\otimes\id\in \mc{W}^1(M)$ and that $\e^{-t\mc{L}}$ is a bounded operator on $\mc{W}^1(M)$ (see \cite[Lem. 3.9]{Schenker:2015}). Thus $\Psi \in \mc{W}^1(M)$ and its Floquet transform 
	$$\pi_\sigma\wh{\Psi}_{\vec k}(x,\omega)
	=\sum_{n\in {\vec p}\Z^d+\sigma} \e^{-{\im}\, {{\vec k}\cdot n} \, }\Psi(x-n, -n, \tau_{n}\omega).$$
	is continuous in $\vec{k}$.
	Direct computation shows that
	\begin{align*}
	\int_{\T^d}\e^{\im {\vec k \cdot y}}\,{\pi_\sigma} \wh{\Psi}_{\vec k}(x,\omega)\,\nu(\di \vec k)
		\ = \ \Psi(x-y,-y,\tau_y\omega) \delta_{\vec{p}\Z + \sigma}(y) .
	\end{align*}
	Thus, by the Fourier-inversion formula,
	$$\sum_{y\in {{\vec p}\Z^d +\sigma}}\e^{-\im {\vec k \cdot y}}\Psi(x-y, -y, \tau_{y}\omega) \ = \ {\pi_\sigma} \wh{\Psi}_{\vec k}(x,\omega), $$
	and
	$$\sum_{y\in {{\vec p}\Z^d+\sigma}}\e^{-\im {\vec k \cdot y}}\Ev{\Psi(x-y, -y, \cdot)} \ = \ {\pi_\sigma} \Ev{\wh{\Psi}_{\vec k}(x,\cdot)}
	\ = \ \ipc{\delta_{x}\otimes\id}{{\pi_\sigma}\wh{\Psi}_{\vec k}}_{L^2(\wh{M};\C)} $$
	for every $\vec{k}\in \T^d$. 
	
	On the other hand, by \eqref{eq:semigFourier}, for $\Phi=\rho_0\otimes\id$, we have 
	$$\wh{\Psi}_{\vec k}=\wh{(\e^{-t\mc{L}}\Phi)}_{\vec k} \ = \ \e^{-t\wc{L}_{{\vec k}}}\, \wh{\Phi}_{{\vec k}},$$
	where 
	\[
	{\pi_\sigma}\wh{\Phi}_{{\vec k}} \ = \ \pi_\sigma\wh{(\rho_0\otimes\id)}_{\vec k}(x,\omega) \
	= \ \sum_{n\in {{\vec p}\Z^d +\sigma}} \e^{-{\im}\, {\vec{\vec k}\cdot n} \, }\psi_0(x-n)\overline{\psi_0(-n)}\otimes\id .
	\]
	Clearly, $\wh{\Phi}_{{\vec k}}=\wh{\rho}_{0;{\vec k}}  \otimes \id$, by the defintion \eqref{eq:rho-hat} of $\wh{\rho}_{0;{\vec k}}$.
	Putting everything together, we have 
	$$\sum_{y\in {{\vec p}\Z^d+\sigma}}\e^{-\im {\vec k \cdot y}}\Ev{\psi_t(x-y) \overline{\psi_t(-y)} }
	\ = \ \ipc{\delta_{x}\otimes\id}{{\pi_\sigma}\,\e^{-t\wc{L}_{{\vec k}}}\, \wh{\rho}_{0;{\vec k}}  \otimes \id}_{L^2(\wh{M};\C)}.$$
	Finally, summing over $\sigma$ in the periodicity cell $\Z_{\vec p}$, we find that 
	\[
	\sum_{y\in \Z^d}\e^{-\im {\vec k \cdot y}}\Ev{\psi_t(x-y) \overline{\psi_t(-y)}}\
		= \ \ipc{{\delta_{x}}\otimes\ora{1}\otimes\id}{\e^{-t\wc{L}_{{\vec k}}}\, \wh{\rho}_{0;{\vec k}}  \otimes \id}_{L^2(\wh{M};\C^{\otimes\vec{p}})}. \qedhere\]
\end{proof}

\section{Spectral analysis on the augmented space}\label{sec:spe-analysis}
\subsection{{Spectral analysis of $\wc{K}_{\vec{0}}$}}\label{sec:KernelHoping}

The spectral analysis of $\wc{L}_{\vec k}$ plays an important role in studying the diffusive scaling of this model. We begin by showing that $0$ is an eigenvalue of $\wc{K}_{0}$. This observation allows us to write down a block decomposition and to find a spectral gap for $\wc{L}_\vec{0}$ in the two sections that follow.

The key observation regarding $\wc{K}_0$ is the following:
 \begin{lem}\label{lem:non-rand-hopp}
	Let $x \in \Z^d$ and $\ora{w} \in \C^{\otimes \vec{p}}$. Then 
	\begin{align} \label{eq:K0-delta0}
	\wc{K}_{\vec 0}\,{\delta_{x}}\otimes\overrightarrow{w}\otimes\id \ = \
	\sum_{\xi\neq\vec 0}h(\xi)\,{\delta_{x-\xi}}\otimes (\rm I-\mc{A}_{\vec{p}}^{-\xi})\ora{w}\otimes\id,
	\end{align}
	where $\mc{A}_{\vec{p}}^\xi=\bigotimes_{j=1}^d (A_{p_j})^{\xi_j}$ with $A_p$ the $p\times p$ right shift matrix,
	\begin{equation}\label{eq:Ap}
	A_{p} := \begin{pmatrix} 0 & 1 & 0 & \cdots & 0 \\
	0 & 0 & 1 & \cdots & 0 \\
	\vdots & \vdots & \ddots & \ddots & \vdots \\
	0 & 0 & 0 &\cdots & 1 \\
	1 & 0 & 0 &\cdots & 0 
	\end{pmatrix}.
	\end{equation} 
\end{lem}
\begin{proof}
	This follows from direct computation: 
	\begin{align*}
	\pi_{\sigma}	\wc{K}_\vec{0} (\delta_x \otimes \ora{w} \otimes \id) &= \sum_{\xi \neq 0} h(\xi)[ \pi_\sigma \delta_{x-\xi} \otimes \ora{w}\otimes \id - \pi_{\sigma-\xi} \delta_{x-\xi}\otimes \ora{w} \otimes \id] \\ 
	&= \sum_{\xi \neq 0} h(\xi) \delta_{x-\xi} \otimes [\pi_{\sigma} - \pi_{\sigma-\xi}]\ora{w} \otimes \id \\
	&= \sum_{\xi \neq 0} h(\xi)\, \delta_{x-\xi} \otimes \pi_{\sigma}(I- \mc{A}_{\vec{p}}^{-\xi})\ora{w} \otimes \id. \qedhere
	\end{align*}
\end{proof}

To proceed we need to consider the matrices $\mc{A}_{\vec{p}}^{\xi}$.  We begin with $A_p$, the $p\times p$ right shift.	
\begin{lem}\label{lem:Aq}
Let $m\in \Z$, $p\in \Z_{>0}$.
The matrix $A_p^m = (A_p)^m$ has $\frac{p}{\gcd(m,p)}$ distinct eigenvalues, 
\begin{equation}
 \e^{2\pi \im \frac{\ell m}{p}}, \ \ \ell = 0, 1, \cdots, \frac{p}{\gcd(m,p)}-1,
\end{equation}
each of multiplicity $\gcd(m,p)$. 
\end{lem}
\begin{proof}Since $A_p^p = \id$, it suffices to restrict our attention to $0<m < p$. The eigenvalues of $A_p$ are all $p$-th roots of unity
\[
	\lambda_\ell = \e^{2\pi \im \frac{\ell}{p}}, \ \ \ell = 0, 1, \cdots, p-1,
\]
and each eigenvalue has multiplicity one.  The corresponding eigenvectors are the elements of the discrete Fourier basis.
For $1 < m < p$, it follows from the spectral mapping theorem that $A_p^m$ has eigenvalues $\lambda_\ell^m$ for $\ell = 0, 1, \cdots, p-1$. From here, it is easy to verify that $\lambda_\ell^m = \lambda_{\ell'}^m$ whenever $|\ell-\ell'| = \frac{np}{\gcd(m,p)}$ for some integer $n$. Finally, since $|\ell - \ell'|< p$, it follows that there are $\frac{p}{\gcd(m,p)}$ distinct eigenvalues each of multiplicity $\gcd(m,p)$.
\end{proof}
This result has an immediate extension to $\mc{A}_{\vec p}$, the tensor product of right shift operators.
\begin{cor}\label{cor:egval-Tensor}
If $\vec{p} = (p_1, \cdots, p_d) \in \Z_{>0}^d$ and $\vec{m} = (m_1, \cdots m_d) \in \Z^d$,  then $\mc{A}_{\vec{p}}^\vec{m}:=\bigotimes_{j=1}^d A_{p_j}^{m_j}$ has eigenvalues 
\begin{equation}
	\prod_{j=1}^d \e^{2\pi \im \frac{\ell_j m_j}{p_j}}; \ \ \ell_j = 0,1, \cdots, \frac{p_j}{\gcd(m_j, p_j)}-1.
\end{equation}
In particular, if $(\vec{e}_j)_{j=1}^d$ is the standard basis on $\Z^d$, then
\begin{equation}\label{eq:kerpej}
	{\rm Ker}(I-\mc{A}_{\vec{p}}^{\vec{e}_j}) = \C^{p_1} \otimes \cdots \otimes \{\ora{1}_{p_j}\} \otimes \cdots \otimes \C^{p_d}.
\end{equation}
\end{cor}

Note that, by eq.\ \eqref{eq:kerpej},
\[\bigcap_{j=1}^d {\rm Ker}(I-\mc{A}_{\vec{p}}^{\vec{e}_j}) \ = \ {\rm span} \{ \ora{1}\}. \]
The following lemma extends this result to a collection $\mc{A}_{\vec{p}}^{\vec{m}_j}$, $j=1,\ldots,k$, where the vectors $\vec{m}_1,\ldots,\vec{m}_k$ generate $\Z^d$.
  
 \begin{lem}\label{lem:tensor-AqWj}
 	Let $\vec{m}_1, \cdots, \vec{m}_k \in \Z^d$, $n_1, \cdots, n_k \in \Z$, and $\vec{M} = n_1\vec{m}_1 + \cdots + n_k \vec{m}_k$ for some $k \ge 1$. Then, we have 
 	\begin{equation}\label{eq:tensor-kerI-M}
 		\bigcap_{j=1}^k {\rm Ker}(I - \mc{A}_{\vec{p}}^{\vec{m}_j}) \subset {\rm Ker}(I - \mc{A}_\vec{p}^{\vec{M}}). 
 	\end{equation}
In particular, if $\vec{m}_1, \cdots ,\vec{m}_k$ generate $\Z^d$, then  	 
 	\begin{align} \label{eq:tensor-CapWj=1}
 	\bigcap_{j=1}^k {\rm Ker}(I - \mc{A}_{\vec{p}}^{\vec{m}_j})=	\bigcap_{j=1}^d {\rm Ker}(I-\mc{A}_\vec{p}^{\vec{e}_j})={\rm span}\{\ora{1}\}.
 	\end{align}
\end{lem}
\begin{proof}
	Suppose $w \in \bigcap_{j=1}^k {\rm Ker}(I - \mc{A}_{\vec{p}}^{\vec{m}_j})$, then for each $j = 1, 2, \cdots, k$,  
	\begin{equation}\label{eq:tensor-EigenValj}
		w = \mc{A}_{\vec{p}}^{\vec{m}_j}w   =  \left(\mc{A}_{\vec{p}}^{\vec{m}_j}\right)^{n_j} w =  \mc{A}_{\vec{p}}^{n_j\vec{m}_j}w.
	\end{equation}
	Repeated application of \eqref{eq:tensor-EigenValj} yields
	\[
		w =  \mc{A}_{\vec{p}}^{n_1\vec{m}_1} = \mc{A}_{\vec{p}}^{n_k\vec{m}_k}w = \mc{A}_{\vec{p}}^{\vec{M}}w.
	\]
	Thus, $w \in {\rm Ker}(I - \mc{A}_{\vec{p}}^{\vec{M}})$. 
	
	If $\vec{m}_1, \cdots, \vec{m}_k$ generate $\Z^d$, then (\ref{eq:tensor-kerI-M}) implies the first equality in (\ref{eq:tensor-CapWj=1}). The second equality follows from Corollary \ref{cor:egval-Tensor} since
	\[
		\bigcap_{j=1}^d {\rm Ker}(I-\mc{A}_{\vec{p}}^{\vec{e}_j}) = \bigcap_{j=1}^d \left( \C^{p_1} \otimes \cdots \otimes \ora{1}_{p_j} \otimes \cdots \otimes \C^{p_d}  \right) = {\rm span}\{\ora{1}\}. \qedhere
	\]

 \end{proof}

We return now to consideration of $\wc{K}_0$. The non-degenerate support condition \eqref{eq:supp-h-nonden} guarantees that the hopping kernel, $h$, is non-zero on a spanning set, $\{\xi_j\}_{j \in J}$, of $\Z^d$. Combining this fact with lemma (\ref{lem:tensor-AqWj}), we can see that $(I-\mc{A}_\vec{p}^{-\xi})\ora{w} = 0$  for all $\xi$ with $h(\xi)\neq 0$ if and only if $\ora{w} \parallel \ora{1}$. In particular, Lemma (\ref{lem:non-rand-hopp}) leads to the following

\begin{cor} \label{cor:KerK0}
	Let $x \in \Z^d$ and $\ora{w} \in \C^{\otimes\vec{p}}$. Then $\wc{K}_0 (\delta_x \otimes \ora{w} \otimes \id) = 0$
if and only if $\ora{w} \parallel \ora{1}$. Moreover, there is $c_0>0$ such that 
for $\ora w\perp\ora 1$, \begin{align}\label{eq:K0delta0-lower}
\norm{\wc{K}_\vec{0} (\delta_x \otimes \ora{w} \otimes \id)}^2 \ \ge \ c_0 \norm{\ora w}^2.
\end{align}
\end{cor}
\begin{proof}
	By Lemma \ref{lem:non-rand-hopp}, we have
	\[ \norm{\wc{K}_\vec{0} (\delta_x \otimes \ora{w} \otimes \id)}^2 \ = \ \sum_{\xi} \abs{h(\xi)}^2 \norm{(I-A_{\vec p}^{-\xi})\ora{w}}^2 . \]
	The right hand side is a quadratic form $Q(w)$ on the finite dimensional space $\C^{\otimes \vec p}$.  Furthermore, 
	by Lemma \ref{lem:non-rand-hopp}, $Q(w)$ vanishes only if $\vec{w}\parallel \ora{1}$.  The lower bound \eqref{eq:K0delta0-lower} follows. In fact, by Lemma \ref{lem:Aq} the smallest eigenvalue of $Q(w)$ on $\{\ora{1}\}^\perp$ is
	\[  c_0 \ = \ \min_{\vec \ell \in \Z_{\vec p}\setminus 0 }  \sum_{\xi} \abs{h(\xi)}^2 \abs{1-\exp\left (-2\pi \im \sum_{j=1}^d \frac{\ell_j\xi_j}{p_j}\right ) }^2 .\]
	Thus $c_0\neq 0$ and eq.\ \eqref{eq:K0delta0-lower} holds.
\end{proof}

\subsection{Block decomposition of $\wc{L}_{\vec 0}$} In the previous section, we showed that $\delta_{\vec{0}} \otimes \ora{1} \otimes \id$ is an eigenvector of $\wc{K}_\vec{0}$ corresponding to the eigenvalue $0$. Using \eqref{eq:U} and \eqref{eq:V}, it is easy to check that this claim also holds  for $\wc{U}$ and $\wc{V}$. Finally, the Markov generator satisfies $B\id=B^{\dagger}\id = 0$. Therefore,
\begin{equation}\label{eq:L0=0}
\wc{L}_{\vec 0} \, \delta_{\vec 0}\otimes\overrightarrow{1}\otimes \id 
=\wc{L}^{\dagger}_{\vec 0} \, \delta_{\vec 0}\otimes\overrightarrow{1}\otimes \id\ = \ 0.
\end{equation}

To further analyze the spectrum of $\wc{L}_{\vec k}$  we will  use  a block decomposition  associated to the following direct sum decomposition of $L^2(\wh{M};\C^{\otimes\vec{p}}) \ \cong \ \ell^2(\Z^d)\otimes \C^{\otimes \vec p}\otimes L^2(\Omega)$:
\begin{equation}\label{eq:directsum}\ell^2(\Z^d)\otimes \C^{\otimes \vec p}\otimes L^2(\Omega) \ = \ \wc{H}_0 \oplus \wc{H}_1 \oplus \wc{H}_2 \oplus \wc{H}_3 ,
\end{equation}
where
\begin{align*}
\wc{H}_{0} \ := 
\operatorname{span} \{  {{\delta_{\vec 0}}}\otimes \overrightarrow{1}\otimes\id\} ,
\end{align*}
\begin{align*}
\wc{H}_{1} \ := \ \delta_{\vec 0}\otimes \{ \ora{1}\}^{\perp}\otimes\id ,
\end{align*}
\begin{align*}
\wc{H}_{2} \ := \ \{\delta_{\vec 0} \}^\perp \otimes \C^{\otimes\vec{p}}\otimes\id \ =\ \ell^2(\Z^d\backslash\{\vec 0\})\otimes \C^{\otimes\vec{p}}\otimes\id ,
\end{align*}
and
\begin{equation*}
\wc{H}_{3}\ := \ \left(\wc{H}_{0} \oplus \wc{H}_{1}\oplus \wc{H}_{2}\right)^{\perp} \ = \
\left\{ \Psi(x,\omega) \ : \ \int_{\Omega} \Psi(x,\omega) \di\mu(\omega) = 0\right\}.
\end{equation*} 
Note that $\dim \wc{H}_0=1$, $\dim \wc{H}_1 = p_1\cdots p_d - 1$, and $\dim \wc{H}_2 = \dim \wc{H}_3 = \infty$.

We will write operators on $L^2(\wh{M};\C^{\otimes\vec{p}})$ as $4\times 4$ matrices of operators acting between the various spaces $\wc{H}_{j},j=0,1,2,3$.  Throughout we will use the notation: 
\begin{enumerate}
	\item $P_j\ =$ the orthogonal projection onto $\wc{H}_j$,
	\item $P_j^\perp \ = \ 1-P_j$.
\end{enumerate}
In particular, $P= P_3^\perp = P_0+P_1+P_2$ is the orthogonal projection of $L^{2}(\wh{M};\C^{\otimes\vec{p}})$ onto the space $\wc{H}_{0} \oplus \wc{H}_{1}\oplus \wc{H}_{2}= \ell^{2}(\Z^d)\otimes\C^{\otimes\vec{p}}\otimes \id$ of ``non-random'' functions:
\begin{equation*}
P\Psi(x) \ = \ \int_{\Omega } \Psi(x,\omega) \di \mu(\omega).
\end{equation*}Then $P_3=P^{\perp}=1-P$ is the projection onto the space of mean zero functions $\wc{H}_3$. 

\begin{lem}\label{lem:block-decomp} On $\wc{H}_0 \oplus \wc{H}_1 \oplus \wc{H}_2 \oplus \wc{H}_3$ the operators $\wc{K}_\vec{0}, \wc{U}, \wc{V}$, and $B$ have following block decomposition
	\begin{align*}
	\wc{K}_{\vec 0} \ &= \ \begin{pmatrix}
	0 & 0 & 0 & 0 \\
	0 & 0 & P_1\wc{K}_{\vec 0}P_2  & 0 \\
	0 & P_2\wc{K}_{\vec 0}P_1 & P_2\wc{K}_{\vec 0}P_2 & 0 \\
	0 & 0 & 0 & P_3\wc{K}_{\vec 0}P_3
	\end{pmatrix},\qquad \wc{U} \ = \ \begin{pmatrix}
	0 & 0 & 0 & 0 \\
	0 & 0 & 0  & 0 \\
	0 & 0 & P_2\wc{U}P_2 & 0 \\
	0 & 0 & 0 & P_3\wc{U}P_3
	\end{pmatrix},\\
	\wc{V}& \ = \ \begin{pmatrix}
	0 & 0 & 0 & 0 \\
	0 & 0 & 0 & 0 \\
	0 & 0 & 0 & P_2\wc{V}P_3  \\
	0 & 0 & P_3\wc{V}P_2 & P_3\wc{V}P_3
	\end{pmatrix}, \quad \text{ and } \quad B \ = \ \begin{pmatrix} 0 & 0 & 0 & 0 \\ 0 & 0 & 0 & 0 \\
	0 & 0 & 0 & 0 \\
	0 & 0 & 0 & P_3BP_3
	\end{pmatrix}.
	\end{align*}
\end{lem}
\begin{proof}
	The eigenvalue equation \eqref{eq:L0=0} gives  
	$$P_0\,\mc{T}=\mc{T}\,P_0=0 $$
	for $\mc{T}=\wc{K}_{\vec 0},\wc U,\wc V,B,\wc{L}_{\vec 0}$. From the definition \eqref{eq:Kk} of $\wc{K}_{\vec 0}$ we see that this operator is ``off-diagonal'' with respect to position, in the sense that $\ipc{\delta_x\otimes F}{\wc{K}_{\vec 0} \delta_x \otimes G}= 0$ for any $x$ and any $F,G\in L^2(\Omega;\C^{\vec p})$.  Thus 
	$P_1\wc{K}_0P_1=0.$ 
	The definitions \eqref{eq:U}, \eqref{eq:V} of $\wc U$, $\wc V$ imply that they vanish on $\delta_{\vec 0}\otimes F$, so
	$$P_1\wc U=P_1\wc V=0,\ \ \wc UP_1=\wc VP_1=0. $$
	
Since $\wc{K}_{\vec 0},\wc U$ are ``non-random'', we have for $j=0,1,2,$ $$P_j\wc{K}_{\vec 0}P_3=0,\ P_3\wc{K}_{\vec 0}P_j=0,\ \ P_j\wc{U}P_3=0,\ P_3\wc{U}P_j=0.$$
Since $\wc V$ is mean zero on $L^2(\Omega)$ and $B\id=B^{\dagger}\id=0$, we have  
\[
	P_3^\perp \wc V P_3^\perp=0,\ \   P_3^\perp B=BP_3^\perp=0 . \qedhere
\]
\end{proof}

\begin{cor}
	On $\wc{H}$ the operator $\wc{L}_{\vec{0}} =  \im \wc{K}_{{\vec 0}} + \im \wc{U} + \im {\lambda} \wc{V}  + B $ has block decomposition
	\begin{align} \label{eq:blockform-L0}
		\wc{L}_{\vec 0}  &=  \begin{pmatrix}
			0 & 0 & 0 & 0 \\
			0 & 0 & \im P_1\wc{K}_{\vec 0}P_2  & 0 \\
			0 & \im P_2\wc{K}_{\vec 0}P_1 & P_2(\im \wc{K}_{\vec 0}+\im \wc{U})P_2 & \im {\lambda} P_2\wc{V}P_3 \\
			0 & 0 & \im {\lambda} P_3\wc{V}P_2 & P_3\wc{L}_{\vec 0}P_3
	\end{pmatrix}.
	\end{align}
\end{cor}
\subsection{Spectral gap}
With the block decomposition \eqref{eq:blockform-L0}, we are now in a position to prove that $\wc{L}_{\vec 0}$ has a spectral gap.  

\begin{lem}\label{lem:L0gap}
		If ${\lambda} >0$, then $0$ is a non-degenerate eigenvalue of $\wc{L}_{\vec 0}$ and there is $g>0$ such that
		\begin{equation*}
		\sigma(\wc{L}_{\vec 0} ) \ = \ \{0\} \cup \Sigma_{+} 
		\end{equation*}
		with $\Sigma_{+} \subset \{z \ : \ \Re z >g\}.$
	For  $\lambda$ small, there is $c=c(\vec p,\|\wh h\|_\infty,\|u\|_\infty,\gamma,T,b)>0$ such that $g\ge c\lambda^2$. 
\end{lem}

\begin{figure}
	\begin{tikzpicture}[scale=1.2]
	\begin{scope}[thick,font=\scriptsize]
	\draw [->] (-1,0) -- (8,0) node [below left]  {$\Re z$};
	\draw [->] (0,-3.2) -- (0,3.2) node [below left] {$\Im z$};
	
	\end{scope}
	\draw [dashed] (0,1) -- (7,3.1);
	\draw [dashed] (0,-1) -- (7,-3.1);
	\filldraw [red] (0,0) circle (0.05);
	
	\node[below,left] at (0,-0.4) {$\sigma(\wc{L}_{\vec 0})\ni 0 $};
	\node [right]at (2,0.7) {$\sigma(\wc{L}_{\vec 0})\backslash\{0\}\subset  \{z:\ \Re z> g\ \}\cap \mc{N}_{+} $};  
	
	\draw [dashed,red] (1.5,2.5) -- (1.5,-2.5);
	
	\draw [dashed,->] (0.5,-2)--(0,-2);
	\draw [dashed,->] (1,-2)--(1.5,-2);
	\node at (0.75,-2) {$g$};
	\draw[fill=red,opacity=0.1]  (1.5,1.45) -- (7,3.1) -- (7,-3.1) -- (1.5,-1.45) -- cycle;
	
	\end{tikzpicture}
	\caption{Spectral gap of $\wc{L}_{\vec 0}$}\label{fig:gap}
\end{figure}

Before proceeding to the proof of the lemma, we note that the sectoriality of $B$ places further restrictions on $\Sigma_+$. Indeed, $\Re \wc{L}_{\vec 0}  = \Re B \ge 0$ in the sense of quadratic forms. Thus, by the sectoriality of $B$, 
	\begin{multline*}
	\abs{\Im{\langle \Phi, \wc{L}_{\vec 0}  \Phi \rangle}} \ \le \ \norm{\wc{K}_ 0  +\wc{U}+ {\lambda} \wc{V}} + \abs{\Im{\langle \Phi, B \Phi \rangle }} 
	\\ \le  \
	\ 2\|\wh h\|_\infty+2\|u\|_\infty+ 2{\lambda} + \gamma \Re \langle \Phi, \wc{L}_{\vec 0}  \Phi \rangle,
	\end{multline*}
	if $\norm{\Phi}=1$. It follows that the numerical range ${\rm Num}(\wc{L}_{\vec 0}) \ = \ \setb{\ipc{\Phi}{\wc{L}_{\vec 0}\Phi}}{\norm{\Phi}=1}$ is contained in
	\begin{equation}\label{eq:numrange}
	\mc{N}_{+} \ := \ \{ z \  : \ \Re z \ge 0 \text{ and } |\Im z| \le 2\|\wh h\|_\infty+2\|u\|_\infty +2 {\lambda} + \gamma \Re z\}.
	\end{equation}
	Since $\sigma(\wc{L}_{\vec 0} ) \subset {\rm Num}(\wc{L}_{\vec 0} )$, we find that $\Sigma_+ \subset \{\Re z> g\}\cap \mc{N}_+$ \textemdash  see Figure \ref{fig:gap}.

To prove Lem.\ \ref{lem:L0gap}, it suffices to show that  the restriction of $\wc{L}_{\vec 0}$ to $\wc{H}_0^\perp=\wc{H}_1 \oplus \wc{H}_2 \oplus \wc{H}_3$,
\begin{align} \label{eq:blockformA}
\mc{J} \ &= \ \begin{pmatrix}
 0 & \im P_1\wc{K}_{\vec 0}P_2  & 0 \\
 \im P_2\wc{K}_{\vec 0}P_1 & P_2(\im \wc{K}_{\vec 0}+\im \wc{U})P_2 & \im {\lambda} P_2\wc{V}P_3 \\
 0 & \im {\lambda} P_3\wc{V}P_2 & P_3\wc{L}_{\vec 0}P_3 
\end{pmatrix},
\end{align}
has spectrum contained in $\{\Re z >g\}$. 
\begin{lem}\label{lem:L0inverse} There is $g>0$, such that whenever $\Re z<g$,
	\begin{enumerate}
	\item $\Gamma_3 -z$ is boundedly invertible on $\wc{H}_3$, where $\Gamma_3= P_3\wc{L}_{{\vec 0}}P_3$,
	\item $\Gamma_2(z)-z$ is boundedly invertible on $\wc{H}_2$, where \begin{equation}\label{eq:Gamma2}\Gamma_2(z) \ = \ P_2\left(\im \wc{K}_{\vec 0}+\im \wc{U}+{\lambda}^2\wc{V}\,(\Gamma_3-z)^{-1}\, \wc{V}\right)P_2,\end{equation} 
	\item $\mc{J}-z$ is boundedly invertible on $\wc{H}_0^\perp$.
	\end{enumerate}
In particular, 
$\mc{J}$ is boundedly invertible.  Let $\Pi_2$ be the projection onto ${\rm Ker}({P_1\wc{K}_0})\subsetneq\wc{H}_2$. 
If $\Pi_2\wt{\phi}\neq0$ for some $\wt{\phi}\in \wc{H}_2$, then $P_2\mc{J}^{-1}\wt{\phi}\neq0$ and 
		\begin{align} \label{eq:gap-gamma}
	\Re \ipc{\wt{\phi}}{{P_2\mc{J}^{-1}\wt{\phi}} }\ge g\|P_2\mc{J}^{-1}\wt{\phi}\|^2>0. 
		\end{align}
	
\end{lem}
\begin{proof}We  obtain this result by repeated applications of the Schur complement formula. As observed above, we may restrict attention in the sectorial domain $z\in \mc{N}_+$.  Fix $z  \in \mc{N}_{+}$ and consider the equation
	\begin{equation}\label{eq:resolvequ}
	( \mc{J}-z )  \begin{pmatrix} \zeta \\ \phi \\
	\Phi \end{pmatrix}  \ 
	= \ \begin{pmatrix}
	-z & \im P_1\wc{K}_{\vec 0}P_2  & 0 \\
	\im P_2\wc{K}_{\vec 0}P_1 & P_2(\im \wc{K}_{\vec 0}+\im \wc{U})P_2-z & \im {\lambda} P_2\wc{V}P_3 \\
	0 & \im {\lambda} P_3\wc{V}P_2 & P_3\wc{L}_{\vec 0}P_3-z
	\end{pmatrix} \begin{pmatrix} \zeta \\ \phi \\
	\Phi \end{pmatrix}	 \ 
	= \ \begin{pmatrix} \wt{\zeta} \\ \wt{\phi} \\
	\wt{\Phi} \end{pmatrix},
	\end{equation}
	for $(\zeta,\phi, \Phi) \in \wc{H}_1 \oplus \wc{H}_2 \oplus \wc{H}_3$ given
	$(\wt{\zeta}, \wt{\phi},
	\wt{\Phi}) \in \wc{H}_1 \oplus \wc{H}_2 \oplus \wc{H}_3$. 
	By the gap condition \eqref{eq:gap} on $B$,
	\begin{equation*}
	\Re P_3\wc{L}_{\vec 0}P_3 \ = \ \Re ( \im \wc{K}_ 0  +\im \wc{U}+ B + \im {\lambda} P_{3} \wc{V} P_{3} ) \ \ge \ \frac{1}{T} P_{3}.
	\end{equation*}
	Therefore, $\Gamma_3-z=P_3\wc{L}_{{\vec 0}}P_3-z$ is boundedly invertible on $\wc{H}_3$ provided $\Re z < \frac{1}{T}$. For such $z$, we may solve the third equation of \eqref{eq:resolvequ} to obtain
	\begin{equation} \label{eq:A3sln}
	\Phi \ = \ (\Gamma_3-z)^{-1}\wt{\Phi}- (\Gamma_3-z)^{-1}\, \im {\lambda} \wc{V} \phi.
	\end{equation}
	
Using the solution \eqref{eq:A3sln}, we  reduce the second equation of \eqref{eq:resolvequ} to
	\begin{equation}
		\left[\Gamma_2(z)-z\right]\, \phi \ = \ \wt{\phi}-\im P_2 \wc{K}_{\vec 0}\zeta -\im {\lambda} P_2\wc{V}\,(\Gamma_3-z)^{-1}\wt{\Phi} \label{eq:A2}
	\end{equation}
	with $\Gamma_2(z)$ as in  \eqref{eq:Gamma2}. For $\varphi\otimes\id \in \wc{H}_2=L^2(\Z^d\backslash\{\vec 0\};\C^{\otimes\vec{p}})$, we have
\begin{multline}
\Re\ipc{\varphi\otimes\id}{\Gamma_2(z)\,\varphi\otimes\id}_{\wc{H}_2}\\
\begin{aligned}
=& \ {{\lambda}^2}\ipc{P_3
	(\Gamma_3-z)^{-1}\wc{V}\,\varphi\otimes\id}
{\left(\Re B-\Re z\right)P_3
	(\Gamma_3-z)^{-1}\wc{V}\,\varphi\otimes\id}_{\wc{H}} \\
\ge& \  {{\lambda}^2}\left (\frac{1}{T}-\Re z \right )\norm
{(\Gamma_3-z)^{-1}\wc{V}\,\varphi\otimes\id}^2_{\wc{H}_3} \\
=& \ {{\lambda}^2}\left (\frac{1}{T}-\Re z \right )\norm
{\left(B^{-1}(\Gamma_3 - z)\right)^{-1}B^{-1}\wc{V}\,\varphi\otimes\id}^2_{\wc{H}_3}
\end{aligned}
\label{eq:437}
\end{multline}
where the inverse of $B$ is well defined since $\wc{V} \varphi\otimes 1 \in \wc{H}_{3} = {\rm Ran} P_{3}$.  Furthermore, $B^{-1}$ is bounded on $\wc{H}_{3}$, with $\norm{B^{-1}P_{3}} \le T$.  Thus $B^{-1} (\Gamma_3 - z)$ is bounded for $z \in \mc{N}_{+}\cap \{ \Re z < \frac{1}{T}\}$ by, 
\begin{align}
\norm{B^{-1}P_3(\Gamma_3 - z)P_3}_{\wc{H}} &\le 1+ \norm{B^{-1}P_3(\wc{K}_{\vec 0}+\wc{U}+{\lambda}\wc{V})}+ |z| \norm{B^{-1}P_3}\nonumber\\
&\le 1+T(2\|\wh h\|_\infty+2\|u\|_\infty +2{\lambda}+|z|)\nonumber\\
&\le 2+\gamma+4T(\|\wh h\|_\infty+\|u\|_\infty + {\lambda}).\label{eq:440}
\end{align}
Putting \eqref{eq:437}, \eqref{eq:440} and \eqref{eq:BinverseV} together, we obtain
\begin{align*}
\Re\ipc{\varphi\otimes\id}{\Gamma_2(z)\,\varphi\otimes\id}_{\wc{H}_2}
\ \ge \ & {{\lambda}^2}\left (\frac{1}{T}-\Re z \right )
\frac{\norm
	{B^{-1}\wc{V}\,\varphi\otimes\id}^2_{\wc{H}} }
{\norm
	{B^{-1}(\Gamma_3 - z)}^2_{\wc{H}} }\\
 \ge \ &  {{\lambda}^2}\left (\frac{1}{T}-\Re z \right )
\frac{\sum\limits_{\sigma\in \Z_{\vec p} } \sum\limits_{x\neq\vec 0} \chi^2|{\pi_\sigma}\varphi(x)|^2 } 
{\left(2+\gamma+4T(\|\wh h\|_\infty+\|u\|_\infty + {\lambda})\right)^2 }  \\
= \ &\frac{ {\lambda}^2\chi^2 (1-T\Re z) } 
{ T \left(2+\gamma+4T(\|\wh h\|_\infty+\|u\|_\infty + {\lambda})\right)^2 }\, \norm{\varphi\otimes\id}^2_{\wc{H}_2}. 
\end{align*}

Let 
\begin{align}\label{eq:c-1}
{ c_1 } \ =\ \frac{{\lambda}^2\chi^2}{T\left ( \lambda^2 \chi^2 + 2\left(2+\gamma+4T(\|\wh h\|_\infty+\|u\|_\infty + {\lambda})\right)^2 \right)},
\end{align}
so that $\frac{ {\lambda}^2\chi^2  } 
{ T \left(2+\gamma+4T(\|\wh h\|_\infty+\|u\|_\infty + {\lambda})\right)^2 }(1 - T c_1) = 2 c_1$.
Then for $z\in \mc{N}_+ \cap \{\Re z \le c_1 \}$, we have 
\begin{align}\label{eq:gapGamma2}
\Re \Gamma_2(z)- \Re z \ge 2{ c_1 } - \Re z \ge c_1,
\end{align}
implying that $\Gamma_2(z)-z$ is boundedly invertible. Thus, \eqref{eq:A2} can be solved on $\wc{H}_2$ to obtain
\begin{multline}
\phi \ = \ (\Gamma_2(z)-z)^{-1}\wt{\phi}-(\Gamma_2(z)-z)^{-1}\im P_2 \wc{K}_{\vec 0}\zeta \\ -(\Gamma_2(z)-z)^{-1}\im {\lambda} P_2\wc{V}\,(\Gamma_3-z)^{-1}\wt{\Phi}. \label{eq:A2sln}
\end{multline}

Now, the first equation of \eqref{eq:resolvequ} reduces to the following
\begin{multline}
 [\Gamma_1(z)-z] \zeta \ = \ \wt{\zeta} -\im P_1\wc{K}_{\vec 0}\,(\Gamma_2(z)-z)^{-1}\wt{\phi} \\ - {\lambda} P_1\wc{K}_{\vec 0}(\Gamma_2(z)-z)^{-1} P_2\wc{V}\,(\Gamma_3-z)^{-1}\wt{\Phi}, \label{eq:A1}
\end{multline}
where $\Gamma_1(z)=P_1\wc{K}_{\vec 0}(\Gamma_2(z)-z)^{-1} P_2\wc{K}_{\vec 0}P_1$. We will use the same strategy to show that $\Gamma_1(z)-z$ is  invertible.
Take $ \zeta= {\delta_{\vec 0}}\otimes \ora{w} \otimes \id\in \wc{H}_1$.   Recall, by definition of $\wc{H}_1$, that $\ora{w}\perp \ora 1$.  Thus, by \eqref{eq:gapGamma2} and Cor. \ref{cor:KerK0},
\begin{multline}\label{eq:gap-est1}
\Re \ipc{\zeta}{\Gamma_1(z)\zeta}_{\wc{H}_1}
 \\ = \ \ipc{(\Gamma_2(z)-z)^{-1}  \wc{K}_{\vec 0}\zeta }{\left(\Re \Gamma_2(z)-\Re z
	\right) (\Gamma_2(z)-z)^{-1}\,  \wc{K}_{\vec 0}\zeta }_{\wc{H}}  
\\ \ge \  \frac{c_1 c_0} 
{\norm{\Gamma_2(z)-z}^2_{\wc{H}}} \norm{ \zeta }^2_{\wc{H}} .
\end{multline} 
For $z \in \mc{N}_{+}\cap \{ \Re z < \frac{1}{2T}\}$,
\begin{multline}\label{eq:gap-est2}
\|\Gamma_2(z)-z\|_{\wc{H}} \ \le \  2\|\wh{h}\|_{\infty}+2\|u\|_{\infty}+4{\lambda}^2\norm{(P_3\wc{L}_{{\vec 0}}P_3-z)^{-1}}_{\wc{H}_3} +|z| \\
\begin{aligned} 
& \le \  4\|\wh{h}\|_{\infty}+4\|u\|_{\infty}+4{\lambda}^2\left(\frac{1}{T}-\frac{1}{2T} \right)^{-1}  +2 {\lambda} + (\gamma+1) \Re z \\
& = \ 4\|\wh{h}\|_{\infty}+4\|u\|_{\infty}+8T{\lambda}^2+2 {\lambda} + (\gamma+1)(2T)^{-1} , 
\end{aligned}
\end{multline}
by \eqref{eq:Gamma2} and \eqref{eq:gap}.  Putting \eqref{eq:gap-est1} and \eqref{eq:gap-est2} together, we obtain
\begin{multline*}
\Re \ipc{\zeta}{\Gamma_1(z)\zeta}_{\wc{H}_1} \\ \ge  \frac{c_1c_0}{\left(4\|\wh{h}\|_{\infty}+4\|u\|_{\infty}+8T{\lambda}^2+2 {\lambda} + (\gamma+1)(2T)^{-1}\right)^2}\,\norm{\zeta}^2_{\wc{H}} \ =: \ {c_2}\norm{\zeta}^2_{\wc{H}}. 
\end{multline*}

Therefore, $\Re \Gamma_1(z)> \Re z$ on $\wc{H}_1$ provided $z\in \mc{N}_+$ and $\Re z< \min\{ c_1 ,\frac{1}{2T}, c_2\}=:g $.  For such $z$ it follows that  $\Gamma_1(z)-z$ is boundedly invertible and \eqref{eq:A1} can be solved on $\wc{H}_1$.
Therefore, \eqref{eq:resolvequ} is explicitly solvable on $\wc{H}=\wc{H}_1 \oplus \wc{H}_2 \oplus \wc{H}_3$ and $ \mc{J}-z $ is boundedly invertible for all $z\in\{z:|\Re z|<g\}\bigcap \mc{N}_+$.

To prove the second part of Lemma \ref{lem:L0inverse}, it is enough to solve $\mc{J}\Psi=\wt \Psi$ for  $\Psi=(\zeta,\phi,\Phi)$ given $\wt{\Psi}=(0,\wt{\phi},0)$. The three equations are reduced to 
\begin{eqnarray*}
\im P_1\wc{K}_{\vec 0}P_2\ \phi&=0 \\
\im P_2\wc{K}_{\vec 0}P_1\ \zeta+ P_2(\im \wc{K}_{\vec 0}+\im \wc{U})P_2\ \phi+ \im {\lambda} P_2\wc{V}P_3\ \Phi&= \wt{\phi} \\
\im {\lambda} P_3\wc{V}P_2\ \phi+ P_3\wc{L}_{\vec 0}P_3\ \Phi&=0
\end{eqnarray*}

   The first equation implies 
$\phi\in {\rm Ker}(P_1\wc{K}_0)$.
Therefore, $\phi=\Pi_2\phi$, where $\Pi_2$ is the projection onto the kernel of $P_1\wc{K}_0$. As derived in the general case, the second and the third equations imply that 
\begin{equation}\label{eq:2+3}
\im P_2 \wc{K}_{\vec 0}\,P_1\zeta+\Gamma_2\phi=\wt{\phi}.
\end{equation}
If $\xi$ satisfies $P_1\wc{K}_{\vec 0}\xi=0$, then $\ipc{\xi}{\wc{K}_{\vec 0}\,P_1\zeta}=\ipc{P_1\wc{K}_{\vec 0}\xi}{\,P_1\zeta}=0$. Therefore, $\Pi_2\wc{K}_{\vec 0}P_1\zeta=0$. Applying $\Pi_2$ to \eqref{eq:2+3}, we have
\begin{align*} \Pi_2\Gamma_2\Pi_2\,\phi=\Pi_2\wt{\phi}.
\end{align*}
Clearly, if $\Pi_2\wt{\phi}\neq0$, then $\phi=P_2\Psi=P_2\mc{J}^{-1}\wt{\phi}\neq0$. 
Notice that $\ipc{ \wc{K}_{\vec 0}\,P_1\zeta}{\phi}=\ipc{ \wc{K}_{\vec 0}\,P_1\zeta}{\Pi_2\phi}=\ipc{\Pi_2 \wc{K}_{\vec 0}\,P_1\zeta}{\Pi_2\phi}=0$. Eq.  
\eqref{eq:2+3} also implies that 
\begin{equation*}
\Re\ipc{\wt{\phi}}{\phi}=\Re\ipc{\im P_2\wc{K}_{\vec 0}\,P_1\zeta+\Gamma_2\phi}{\phi}
=\Re\ipc{\Gamma_2\phi}{\phi}\ge 2c_1\,\|\phi\|^2\ge g\,\|\phi\|^2>0, 
\end{equation*}
which completes the proof of \eqref{eq:gap-gamma}. \qedhere
\end{proof}
 The spectral gap ${g }$ of $\wc{L}_{{\vec 0}}$ has consequences for the dynamics of the semi-group. 
\begin{lem} \label{lem:L0dynamics} Let $Q_{\vec 0}=$ orthogonal projection onto $\wc{H}_0={\rm span}\, \delta_{\vec 0}\otimes \overrightarrow{1} \otimes 1$  in $L^{2}(\wh{M};\C^{\otimes\vec{p}})$.  Then $\e^{-t \wc{L}_{{\vec 0}}} (1-Q_{\vec 0})$
	is a contraction semi-group on $\ran (1-Q_{\vec 0})$, and for all sufficiently small $\epsilon >0$  there is $C_{\epsilon} > 0$ such that
	\begin{equation}\label{eq:expdecay}
	\norm{\e^{-t \wc{L}_{{\vec 0}}} (1- Q_{\vec 0})}_{L^{2}(\wh{M};\C^{\otimes\vec{p}})} \ \le \ C_{\epsilon} \e^{-t ({g } -\epsilon)}
	\end{equation}	
\end{lem}

\begin{lem}\label{lem:Lkgap}
	There is $c_0>0$ such that 
	\begin{equation*}
	\norm{\wc{L}_{{\vec k}} - \wc{L}_{\vec{0}}}_{L^{2}(\wh{M};\C^{\otimes\vec{p}})} \ \le \ c_0 |{\vec k}|.
	\end{equation*}
	
	If  $|{\vec k}|$ is sufficiently small, the spectrum of $\wc{L}_{{\vec k}}$ consists of:
	\begin{enumerate}
		\item A non-degenerate eigenvalue $E({\vec k})$ contained in $S_0 = \{ z\, :\, |z| < c_0 |{\vec k}|\}.$ 
		\item The rest of the spectrum is contained in the half plane $ S_1 = \{ z : \Re z > {g } - c_0 |{\vec k}| \}$ such that $S_{0} \cap S_1 = \emptyset$. 
	\end{enumerate}	
Furthermore,	{
	$E({\vec k})$ is $C^{2}$ in a neighborhood of $\vec 0$,
	\begin{equation} \label{eq:E0E'0}
	E({\vec 0}) = 0, \quad \grad E({\vec 0}) = 0. 
	\end{equation}
	Denote $\partial_j=\partial_{k_j}$  and $\varphi_{\vec 0}=\frac{1}{\sqrt{ \otimes \vec p }}\delta_{\vec 0} \otimes \overrightarrow{1}\otimes\id$ for simplicity where $ \otimes \vec p =p_1\cdot p_2\cdots p_d$, then 		
		\begin{align}\label{eq:Eij}
		\partial_i\partial_jE({0}) \ =
		\ipc{\partial_j\wc{K}_{\vec 0}\varphi_{\vec 0}}{{ P_2\mc{J}^{-1}  }\,  \partial_i  \wc{K}_ {\vec 0} \varphi_{\vec 0} } +\ipc{\partial_i\wc{K}_{\vec 0}\varphi_{\vec 0}}{{ P_2\mc{J}^{-1}  }\,  \partial_j  \wc{K}_ {\vec 0} \varphi_{\vec 0} }
		\end{align}
		where 
$P_2,\mc{J}$ and $\mc{J}^{-1}$ are given  in \eqref{eq:blockformA} and Lemma \ref{lem:L0inverse}.  
	}	 
\end{lem}

	\begin{remark}\label{rem:ReD}
		Let $\vec D:=\left(\vec{D}_{i,j}\right)_{d\times d}=\left(\partial_i\partial_jE(\vec 0)\right)_{d\times d}$. It is clear from \eqref{eq:Eij} that $\vec{D}$ is symmetric.   Furthermore, for any $\vec k\in\T^d$, in view of the expression of $\partial_i\wc{K}_{\vec 0}$ in \eqref{eq:K0'-delta0}, $0\neq\sum_{i}  k_i \partial_i\wc{K}_{\vec 0}\varphi_{\vec 0}\in \ell^2(\Z^d)\otimes \ora{1}\otimes\id$. It is non-zero due the non-degeneracy of $h$. Therfore, by \eqref{eq:gap-gamma} in Lemma \ref{lem:L0inverse},
		$$\Re \ipc{\vec{k}}{\vec{D} \vec{k}} \ 
		= \ 2  \Re \ipc{   \sum_{i}  k_i \partial_i\wc{K}_{\vec 0}\varphi_{\vec 0} } 
		{ { P_2\mc{J}^{-1}  }\sum_{i}  k_i \partial_i\wc{K}_{\vec 0}\varphi_{\vec 0} } 
		\ > \ 2g\norm{ { P_2\mc{J}^{-1}  }\sum_{i}  k_i \partial_i\wc{K}_{\vec 0}\varphi_{\vec 0}}^2>0.$$
	  	In the next section, we will relate the matrix element of $\vec D$ with limits of diffusively scaled moments. From the real valued moments, we will see that $\partial_i\partial_jE(\vec 0)\in\R$ and then $\vec D$ is positive definite. 
\end{remark}

Similar to Lemma \ref{lem:L0dynamics}, dynamical information about the semi-group $\e^{-t \wc{L}_{{\vec k}}}$ follows from the spectral gap of $\wc{L}_{\vec k}$ in Lemma \ref{lem:Lkgap}:
\begin{lem}\label{lem:Lkdynamics} If $\epsilon$ is sufficiently small, then there is $C_{\epsilon } <\infty$ such that
	\begin{equation*}
	\norm{\e^{-t \wc{L}_{{\vec k}}}(1- Q_{{\vec k}})}_{L^{2}(\wh{M};\C^{\otimes\vec{p}})} \le C_{\epsilon}\e^{-t ({g } - \epsilon - c_0 |{\vec k}|)}
	\end{equation*}
	for all sufficiently small ${\vec k}$.
\end{lem}

Notice that $ \otimes \vec p =p_1\cdot p_2\cdots   p_d$. The case where $d=1$ and $\otimes \vec p=p_1=1$ is equivalent to the free case considered in \cite{KS2009}, where the above lemmas were proved. { The proof follows from the standard perturbation theory of analytic semi-groups\textemdash see for instance \cite{Engel:2000nx, Kato:1995sf}. There are no essential differences in the proof when $\otimes \vec p>1$. We omit the proofs for Lemma \ref{lem:L0dynamics}-Lemma \ref{lem:Lkdynamics} here. We only sketch the proofs for \eqref{eq:E0E'0} and \eqref{eq:Eij}, which plays the most important role for the explicit expression of the diffusion constant in the next section.}

\begin{proof} [Proof of {\eqref{eq:E0E'0} and \eqref{eq:Eij}}] Write $\partial_j=\partial_{k_j}$ for short. 
	Let $E({\vec k})$  be the non-degenerate eigenvalue of $\wc{L}_{\vec k}$, and the associated normalized eigenvector $\varphi_{{\vec k}}$. Let $Q_{\vec k}$ be the orthogonal projection onto $\varphi_{{\vec k}}$.  Clearly $E(\vec 0)=0$, $\varphi_{\vec 0}=\frac{1}{\sqrt  \otimes \vec p } \delta_{\vec 0}\otimes \overrightarrow{1}\otimes \id$ and $\wc{L}_{{\vec 0}}\varphi_{\vec 0}=\wc{L}^\dagger_{\vec 0}\varphi_{\vec 0}=0$. Since 
	\begin{equation}
	\wc{L}_{{\vec k}}\varphi_{{\vec k}} \ = \ E({\vec k}) \varphi_{\vec k},
	\end{equation}
	direct computation shows 
	\begin{align}
	&\partial_j \wc{L}_{{\vec k}}\,\varphi_{{\vec k}}+ \wc{L}_{{\vec k}}\partial_j \varphi_{\vec k}\ = \ \partial_j  E({\vec k}) \varphi_{{\vec k}}+E({\vec k}) \partial_j \varphi_{\vec k}\label{eq:Lk=Ek'}\\
\Longrightarrow&{\partial_j}\wc{L}_{\vec{ 0 }}\varphi_{\vec { 0 }}+ \wc{L}_{{ 0 }}{\partial_j}\varphi_{\vec  0 }\ = \ {\partial_j} E({\vec  0 }) \varphi_{\vec { 0 }}.  \label{eq:L0=E0'}
	\end{align}
Notice that $\partial_j \wc{L}_{{\vec 0}}=\im \partial_j \wc{K}_{\vec 0}$ maps $\wc{H}_0=\ran Q_{\vec 0}$ to $\wc{H}_2$, therefore, $Q_{\vec 0}\partial_j \wc{L}_{{\vec 0}}=0$ and 
\begin{align*}
 \partial_j  E(\vec 0)=\ipc{\varphi_{\vec 0}}{  \partial_j  \wc{L}_ {\vec 0} \varphi_{\vec 0} }+\ipc{\varphi_{\vec 0}}{\wc{L}_{\vec 0}   \partial_j  \varphi_ {\vec 0} } 
=\ipc{Q_{\vec 0}\varphi_{\vec 0}}{  \partial_j  \wc{L}_ {\vec 0} \varphi_{\vec 0} }+\ipc{\wc{L}^\dagger_{\vec 0} \varphi_{\vec 0}}{  \partial_j  \varphi_ {\vec 0} }
=0.
\end{align*}
Differentiating \eqref{eq:Lk=Ek'} again, we have 
\begin{align}
&\partial_i\partial_j\wc{L}_{{\vec k}}\varphi_{{\vec k}}+\partial_j\wc{L}_{{\vec k}}\partial_i\varphi_{\vec k}+\partial_i\wc{L}_{{\vec k}}\partial_j\varphi_{\vec k}+\wc{L}_{{\vec k}}\partial_i\partial_j\varphi_{{\vec k}}  \nonumber \\
=& \partial_i\partial_jE({\vec k}) \varphi_{{\vec k}}+\partial_j E({\vec k}) \partial_i\varphi_{\vec k}
+\partial_i E({\vec k}) \partial_j\varphi_{\vec k}+E({\vec k}) \partial_i\partial_j\varphi_{{\vec k}}.\label{eq:Lk=Ek''}
\end{align}
Evaluating \eqref{eq:Lk=Ek''} at $\vec \vec k=\vec 0 $ and using $\grad E(\vec 0)=0$, we have that 
\begin{align*}
&\partial_i\partial_j\wc{L}_{{{0}}}\varphi_{\vec {{0}}}+\partial_j\wc{L}_{\vec{{0}}}\partial_i\varphi_{\vec {0}}+\partial_i\wc{L}_{\vec{{0}}}\partial_j\varphi_{\vec {0}}+\wc{L}_{{{0}}}\partial_i\partial_j\varphi_{\vec 0}= \partial_i\partial_jE({{\vec 0}}) \varphi_{\vec 0}.
\end{align*}
We also have $Q_{\vec 0}\partial_i\partial_j\wc{L}_{\vec 0}=0$ for the same reason as for $ Q_0\partial_j  \wc{L}_{{\vec 0}}$. Notice that  $  \partial_j  \wc{L}_ {\vec 0}=\im  \partial_j  \wc{K}_{\vec 0}=-\partial_j \wc{L} ^\dagger_ 0$ 
and $  \partial_j  \wc{L}_ {\vec 0} \varphi_{\vec 0}\in \ell^{2}(\Z^d\backslash \{\vec 0 \})\otimes \ora 1\otimes \id$ because of \eqref{eq:K0'-delta0}. Corollary \ref{cor:KerK0} implies $ \partial_j  \wc{L}_ {\vec 0} \varphi_{\vec 0}\in {\rm Ker}(P_1\wc{K}_{\vec 0})=\ran (\Pi_2)\subsetneq \wc{H}_2$.  Therefore, 
\begin{align*}
\partial_j\partial_jE({{\vec 0}})
 =&\ipc{\varphi_{\vec 0}}{\partial_j\wc{L}_{\vec{{0}}}\partial_i\varphi_{\vec {0}} }+
 \ipc{\varphi_{\vec 0}}{\partial_i\wc{L}_{\vec{{0}}}\partial_j\varphi_{\vec {0}} }\\
  =&{\rm i}\ipc{\partial_j\wc{K}_{\vec 0 }\,\varphi_{\vec 0}}{\partial_i\varphi_{\vec {0}} }+{\rm i}
 \ipc{\partial_i\wc{K}_{\vec 0}\, \varphi_{\vec 0}}{\partial_j\varphi_{\vec {0}} } \\
 =&{\rm i}\ipc{P_2\partial_j\wc{K}_{\vec 0 }\,\varphi_{\vec 0}}{P_2\partial_i\varphi_{\vec {0}} }+{\rm i}
 \ipc{P_2\partial_i\wc{K}_{\vec 0}\, \varphi_{\vec 0}}{P_2\partial_j\varphi_{\vec {0}} } .
\end{align*}

It remains to solve 
$${\partial_j}\wc{L}_{\vec{ 0 }}\varphi_{\vec { 0 }}+ \wc{L}_{{ 0 }}{\partial_j}\varphi_{\vec  0 }\ =0$$
i.e.,
\begin{equation}\label{eq:partial}
{\rm i}{\partial_j}\wc{K}_{\vec{ 0 }}\varphi_{\vec { 0 }}+ \wc{L}_{{ 0 }}{\partial_j}\varphi_{\vec  0 }\ =0
\end{equation}
 for $P_2\partial_i\varphi_{\vec {0}}$. Recall the block form of $\wc{L}_{{\vec 0}}$ in \eqref{eq:blockform-L0} and $\mc{J}$ in \eqref{eq:blockformA}.  The key fact $\partial_j  \wc{K}_ {\vec 0} \varphi_{\vec 0}=\Pi_2\partial_j  \wc{K}_ {\vec 0} \varphi_{\vec 0}\in \wc{H}_2$ reduces eq. \eqref{eq:partial} to what we have considered in the second part of Lemma \ref{lem:L0inverse}: 
\begin{align*}
\begin{pmatrix}
0 & 0 & 0 & 0 \\
0 & 0 & \im P_1\wc{K}_{\vec 0}P_2  & 0 \\
0 & \im P_2\wc{K}_{\vec 0}P_1 & P_2(\im \wc{K}_{\vec 0}+\im \wc{U})P_2 & \im {\lambda} P_2\wc{V}P_3 \\
0 & 0 & \im {\lambda} P_3\wc{V}P_2 & P_3\wc{L}_{\vec 0}P_3
\end{pmatrix}\, 
\begin{pmatrix}
\ast \\
\ast \\
P_2 \partial_j \varphi_{\vec 0}\\
\ast
\end{pmatrix}=\begin{pmatrix}
0 \\
0 \\
-  {\rm i}\partial_j  \wc{K}_ {\vec 0} \varphi_{\vec 0}\\
0
\end{pmatrix}. 
\end{align*}
As derived in Lemma \ref{lem:L0inverse}:
	\begin{align*}
P_2  \partial_j  \varphi_ {\vec 0}
=-{\rm i}P_2\mc{J}^{-1}\,\partial_j  \wc{K}_ {\vec 0} \varphi_{\vec 0},
	\end{align*}
	where $P_2$ is the projection onto $\wc{H}_2$. 
 Therefore, 
	\begin{align*}
	\partial_j\partial_jE({{\vec 0}})=&{\rm i}\ipc{\partial_j\wc{K}_{\vec 0 }\,\varphi_{\vec 0}}{  -{\rm i}P_2\mc{J}^{-1}\,\partial_i  \wc{K}_ {\vec 0} \varphi_{\vec 0} }+{\rm i}
	\ipc{\partial_i\wc{K}_{\vec 0}\, \varphi_{\vec 0}}{ -{\rm i}P_2\mc{J}^{-1}\,\partial_j  \wc{K}_ {\vec 0} \varphi_{\vec 0} } \\
	=&\ipc{\partial_j\wc{K}_{\vec 0 }\,\varphi_{\vec 0}}{ P_2\mc{J}^{-1}\,\partial_i  \wc{K}_ {\vec 0} \varphi_{\vec 0} }+
	\ipc{\partial_i\wc{K}_{\vec 0}\, \varphi_{\vec 0}}{ P_2\mc{J}^{-1}\,\partial_j  \wc{K}_ {\vec 0} \varphi_{\vec 0} } ,
	\end{align*}
which gives \eqref{eq:Eij}. 
\end{proof}

\section{Proof of the main results}\label{sec:MainResults}
\subsection{Central limit theorem}
We first prove  \eqref{eq:genCLT} for bounded continuous $f$ and normalized $\psi_0\in\ell^2(\Z^d)$.  The extension to quadratically bounded $f$  follows from  some standard arguments combining \eqref{eq:genCLT} for bounded continuous $f$ and diffusive scaling for second moments, Lemma \ref{lem:lim-Mij}.  We refer readers to Sec.\ 4.5 in \cite{Schenker:2015} for more details about this extension. We omit the proof of the extension here. 

To prove  \eqref{eq:genCLT} for bounded continuous $f$, it suffices, by Levy's Continuity Theorem and a limiting argument, to prove 
\begin{equation}\label{eq:CLT-1}
\lim_{t\rightarrow \infty} \sum_{x\in \Z^d} \e^{\im {\vec k}\cdot   \frac{x}{\sqrt{t}}} \Ev{|\psi_t(x)|^2} \ = \ \e^{-\frac{1}{2}\ipc{\vec k}{\vec D\vec k}},
\end{equation}
where $\psi_t(x)\in\ell^2({\Z^d})$ is the solution to eq.\ \eqref{eq:genE} with initial condition $\psi_0\in\ell^2(\Z^d)$.  As pointed out in Sec.\ 4.2, \cite{Schenker:2015},  it is enough to establish eq.\ \eqref{eq:CLT-1} for $\psi_0\in\ell^1(\Z^d)$; it then extends to all of $\psi_0\in\ell^2(\Z^d)$ by a limiting argument. So throughout this section, we assume that 
	\begin{align}\label{eq:inital-l1}
	\norm{\psi_0}_{\ell^2}=1,\ \ \ {\rm and }\ \ \  \norm{\psi_0}_{\ell^1}:=\sum_{x\in\Z^d}|\psi_0(x)|<\infty. 
	\end{align}    
We also denote for simplicity
\begin{align}\label{eq:simp-notation}
{\varphi}_{\vec 0}:={\varphi}_{\vec 0}(x,\omega)=\frac{1}{\sqrt{\otimes \vec p}}{\delta_{\vec 0}}\otimes \overrightarrow{1} \otimes\id,\ {\Phi}_{\vec k}:={\Phi}_{\vec k}(x,\omega)=\sqrt{\otimes \vec p}\cdot\wh{\rho}_{0;\vec k}(x)\otimes\id,
\end{align}
where $\overrightarrow{1},\wh{\rho}_{0;\vec k}(x)\in\C^{\otimes\vec{p}}$ are defined  in \eqref{eq:rho-hat}. Recall that for any $\sigma\in\Z_{\vec p}$
\begin{align*}
\pi_\sigma\overrightarrow{1}=1,\qquad  \pi_\sigma\wh{\rho}_{0;\vec 0}(x)=
\sum\limits_{n\in {\vec p}\Z^d+\sigma}   \psi_0(x-n) \overline{\psi_0(-n)}.
\end{align*}
By \eqref{eq:Fourier-PFK}, we have
\begin{equation*}
\sum_{x\in\Z^d } \e^{\im \frac{\vec k}{\sqrt t} x } \Ev{|\psi_t(x)|^2 } \ = \ \ipc{ {\varphi}_{\vec 0}\ }{\e^{-t \wc{L}_{\nicefrac{\vec k}{\sqrt t}}} \, {\Phi}_{\frac{\vec k}{\sqrt t}} }_{L^2(\wh{M};\C^{\otimes\vec{p}})}. \end{equation*}

Letting $Q_{{\vec k}}$ denote the Riesz projection onto the eigenvector of $\wc{L}_{{\vec k}}$  near zero, we have

\begin{align}
\sum_{x\in\Z^d } \e^{\im \frac{\vec k}{\sqrt t} x } \Ev{|\psi_t(x)|^2 }=&
\ipc{ {\varphi}_{\vec 0}\ }{\e^{-t \wc{L}_{\nicefrac{\vec k}{\sqrt t}}} Q_{\frac{\vec k}{\sqrt t}} {\Phi}_{\frac{\vec k}{\sqrt t}} }  +
\ipc{ {\varphi}_{\vec 0}\ }{\e^{-t \wc{L}_{\nicefrac{\vec k}{\sqrt t}}} \left(1-Q_{\frac{\vec k}{\sqrt t}}\right) {\Phi}_{\frac{\vec k}{\sqrt t}} } \nonumber \\
=&\e^{-t E({\frac{\vec k}{\sqrt t}})}\ipc{ {\varphi}_{\vec 0}\ }{ Q_{\frac{\vec k}{\sqrt t}} {\Phi}_{\frac{\vec k}{\sqrt t}} } +
\ipc{ {\varphi}_{\vec 0}\ }{\e^{-t \wc{L}_{\nicefrac{\vec k}{\sqrt t}}} \left(1-Q_{\frac{\vec k}{\sqrt t}}\right) {\Phi}_{\frac{\vec k}{\sqrt t}} }. \label{eq:1-Q-pf}
\end{align}
By Lemma \ref{lem:Lkdynamics}, 
the second term in \eqref{eq:1-Q-pf} is exponentially small in the large $t$ limit,
\begin{align}
\abs{\ipc{ {\varphi}_{\vec 0}\ }{\e^{-t \wc{L}_{\nicefrac{\vec k}{\sqrt t}}} \left(1-Q_{\frac{\vec k}{\sqrt t}}\right) {\Phi}_{\frac{\vec k}{\sqrt t}} }}\le& 
  \  \norm{(1 - Q_{\frac{\vec k}{\sqrt{t}}}) \e^{- t  \wc{L}_{ \nicefrac{\vec k}{\sqrt{t}}}}}\cdot \norm{{\varphi}_{\vec 0}} \cdot  \norm{  {\Phi}_{\frac{\vec k}{\sqrt t}}} \nonumber \\
  \le & \  C_{\epsilon}  \e^{- t  ({g }-\epsilon - c \frac{|\vec k|}{\sqrt t})} \cdot \norm{{\varphi}_{\vec 0}} \cdot  \norm{  {\Phi}_{\frac{\vec k}{\sqrt t}}}. \label{eq:decay-pf}
\end{align}
Direct computation shows that 
	\begin{align*}
	\lim_{t\to\infty}\norm{  {\Phi}_{\frac{\vec k}{\sqrt t}}}^2_{L^2(\wh{M};\C^{\otimes\vec{p}})}
	=&(\otimes \vec p)\norm{\wh{\rho}_{0;\vec 0}}^2_{\ell^2(\Z^d;\C^{\otimes\vec{p}})}\\
	\le &(\otimes \vec p) \sum_{\sigma\in \Z_{\vec p} }\sum_{x\in\Z^d}\left|{\sum_{ n\in {\vec p}\Z^d+\sigma }   \psi_0(x-n) \overline{\psi_0(-n)}}\right|^2 \\
		\le & (\otimes \vec p)\norm{\psi_0}^2_{\ell^2}\sum_{\sigma\in \Z_{\vec p} }\left(\sum_{ n\in {\vec p}\Z^d+\sigma }\left| {\psi_0(-n)}\right|\right)^{2}\\
		 	\le & (\otimes \vec p)\norm{\psi_0}^2_{\ell^2}\cdot \norm{\psi_0}^2_{\ell^1}<\infty .
	\end{align*}
Therefore, in \eqref{eq:decay-pf},  $\abs{\ipc{ {\varphi}_{\vec 0}\ }{\e^{-t \wc{L}_{\frac{\vec k}{\sqrt t}}} \left(1-Q_{\frac{\vec k}{\sqrt t}}\right) {\Phi}_{\frac{\vec k}{\sqrt t}} }}\longrightarrow 0$ as $t\to\infty$. 

Regarding the first term in \eqref{eq:1-Q-pf}, we have by Taylor's formula, 
\begin{equation*}
E\left(\frac{\vec k}{\sqrt t}\right) \ = \ \frac{1}{2}\sum_{i,j}\partial_j\partial_j E(\vec 0)\frac{k_i}{\sqrt t}\frac{k_j}{\sqrt t} \ + \ o(\frac{1}{t})=\ \frac{1}{2t}\sum_{i,j}\partial_j\partial_j E(\vec 0){k_i}{k_j} \ + \ o(\frac{1}{t}),
\end{equation*}
since $E(\vec 0)=\grad E(\vec 0)=0$. Thus, 
\begin{equation}\label{eq:taylor}
\e^{- t E({\vec k}/\sqrt{t})} \ = \ \e^{-t  \, \frac{1}{2t}\sum_{i,j}\partial_j\partial_j E(\vec 0){k_i}{k_j}} + o(1)=\e^{-\frac{1}{2}\sum_{i,j}\partial_j\partial_j E(\vec 0){k_i}{k_j}} + o(1).
\end{equation}
Direct compuation shows that 
\begin{align*}
\ipc{{\varphi}_{\vec 0}}{{\Phi}_{\vec 0}}_{L^2({\wh{M};\C^{\otimes\vec{p}}})}
=&\ipc{ {\delta_{\vec 0}}\otimes \overrightarrow{1} \otimes\id\ }{ \wh{\rho}_{0;\vec 0}\otimes\id }_{L^2({\wh{M};\C^{\otimes\vec{p}}})} \\
=&  \sum_{\sigma\in \Z_{\vec p} }\sum_{ n\in {\vec p}\Z^d+\sigma }   \psi_0(-n) \overline{\psi_0(-n)}=  \norm{\psi_0}^2_{\ell^2(\Z^d;\C)}=1. 
\end{align*}
Thus, 
\begin{align}\label{eq:Q0Phi_0}
Q_{\vec 0}{\Phi}_{\vec 0}={\rm Proj}_{{\varphi}_{\vec 0}}{\Phi}_{\vec 0}
= \ipc{{\varphi}_{\vec 0}}{{\Phi}_{\vec 0}}\cdot \frac{{\varphi}_{\vec 0}}{\norm{{\varphi}_{\vec 0}}^2}= {\varphi}_{\vec 0}.
\end{align}
Putting together everything, we have 
\begin{align*}
 \lim_{t\to\infty} \sum_{x\in\Z^d } \e^{\im \frac{\vec k}{\sqrt t} x } \Ev{|\psi_t(x)|^2 }=&
 \lim_{t\to\infty} \e^{-t E\left({\frac{\vec k}{\sqrt t}}\right)}\ipc{ {\varphi}_{\vec 0}\ }{ Q_{\frac{\vec k}{\sqrt t}} {\Phi}_{\frac{\vec k}{\sqrt t}} }\\
 =& \e^{-\frac{1}{2}\sum_{i,j}\partial_j\partial_j E(\vec 0){k_i}{k_j}} \ipc{ {\varphi}_{\vec 0}\ }{ Q_{\vec 0} {\Phi}_{\vec 0} } 
  = \e^{-\frac{1}{2}\sum_{i,j}\partial_j\partial_j E(\vec 0){k_i}{k_j}}. 
\end{align*}
Therefore, \eqref{eq:CLT-1} holds true with $\vec D_{i,j}=\partial_j\partial_j E(\vec 0)$ for any normalized $\psi_{0}\in\ell^2(\Z^d)$. \qed

\subsection{Diffusive scaling and reality of the diffusion matrix}
We proceed to prove the diffusive scaling \eqref{eq:genDS} under the assumption that 
\begin{align}
\sum_{x} \abs{\psi_0(x)}^2=1,\ \ \sum_{x} |x|^2 \abs{\psi_0(x)}^2 < \infty. \label{eq:initial}
\end{align}
 Similar to \eqref{eq:inital-l1}, it is enough to establish the results for $x\psi_0\in\ell^1(\Z^d)$; it then extends to all of $x\psi_0\in\ell^2(\Z^d)$ by a limiting argument. We assume that 
 \begin{align}\label{eq:initial-xl1}
\sum_{x} |x| \abs{\psi_0(x)} < \infty.
 \end{align}
We continue to use the notation in \eqref{eq:simp-notation}. Also, $\ipc{\cdot}{\cdot}$ will stand for $\ipc{\cdot}{\cdot}_{L^2(\wh{M};\C^{\otimes \vec p})}$ unless otherwise specified. We also denote $\partial_i=\partial_{k_i},i=1,\cdots,d$ for short. 

As pointed out in Sec. 4.4 in \cite{Schenker:2015}, $\sum_{x} (1 + |x|^2) |\psi_t(x)|^2 \le \e^{Ct}$ for each $t >0$.  Thus the second moments of the position
\begin{align} \label{eq:Mij}
M_{i,j}(t) \ := \  \sum_{x\in \Z^d} x_i x_j \Ev{\abs{\psi_t(x)}^2}
\end{align}
are well defined and finite. The main task of this section is to show that $M_{i,j}(t) \sim \vec{D}_{i,j}t$, where $\vec{D}_{i,j}=\partial_i\partial_jE(\vec 0)$ are given in \eqref{eq:Eij}.  More precisely, 
\begin{lem} \label{lem:lim-Mij}
Let $P_2\mc{J}^{-1}$ be as in Lemma \ref{lem:L0inverse} . Suppose the initial value $\psi_0$ satisfies \eqref{eq:initial}, then for all $1\le i,j\le d$, 
	\begin{align*}
	\lim_{t\to\infty}\frac{1}{t}M_{i,j}(t)=&\ipc{\partial_j\wc{K}_{\vec 0}\varphi_{\vec 0}}{P_2\mc{J}^{-1}\,  \partial_i  \wc{K}_ {\vec 0} \varphi_{\vec 0} } 
	+\ipc{\partial_i\wc{K}_{\vec 0}\varphi_{\vec 0}}{P_2\mc{J}^{-1}\,  \partial_j  \wc{K}_ {\vec 0} \varphi_{\vec 0} } 
	=\partial_i\partial_jE(\vec 0).
	\end{align*}
	As a consequence, $\partial_i\partial_jE(\vec 0)\in\R$ and ${\vec D}=\left(\partial_i\partial_jE(\vec 0)\right)_{d\times d}$ is positive definite. In particular, 
	\begin{align*}
	\lim_{t\to\infty}\frac{1}{t}\sum_{x\in \Z^d} |x|^2\,\Ev{\abs{\psi_t(x)}^2}
	=2\sum_{i=1}^d \ipc{\partial_i\wc{K}_{\vec 0}\varphi_{\vec 0}}{P_2\mc{J}^{-1}\,  \partial_i  \wc{K}_ {\vec 0} \varphi_{\vec 0} } 
	=\tr {\vec D}\in(0,\infty).
	\end{align*}
\end{lem}
By \eqref{eq:Fourier-PFK}, we have 
\begin{align}\label{eq:Mij-1}
M_{i,j} (t) \ = \ - \left .  \partial_{i} \partial_{j} \sum_{x\in \Z^d} \e^{\im \vec{k}\cdot x}  \Ev{\abs{\psi_t(x)}^2} \right |_{\vec{k}=\vec 0}=\ - \left .  \partial_{i} \partial_{j} \ipc{{\varphi}_{\vec 0} }{\e^{-t \wc{L}_{{\vec k}}} \,{\Phi}_{\vec k} }\right |_{\vec{k}=\vec 0}.
\end{align}
The following decomposition of $M_{i,j}$ are essentially contained in \cite{Schenker:2015}. We sketch the proof in Appendix \ref{app:A} for readers' convenience. 
\begin{lem}\label{lem:M=N1-5}
	For all $1\le i,j\le d$ and $t\in\R^+$, $M_{i,j}=\sum\limits_{n=1}\limits^5N_n$, where 
	\begin{align}
	N_1=&-\ipc{ {\varphi}_{\vec 0}\ }{ \partial_{i} \partial_{j} {\Phi}_{\vec 0} }; \label{eq:N1}\\
	N_2=&\int_{ 0}^{t}\left[\ipc{\partial_i \wc{L} ^\dagger_{\vec 0} {\varphi}_{\vec 0}\ }{	\, \e^{-s \wc{L}_{\vec 0}}\, (1-Q_{\vec 0}) \,   \partial_j {\Phi}_{\vec 0} }\, 
	     +\ipc{\partial_j \wc{L} ^\dagger_{\vec 0} {\varphi}_{\vec 0}\ }{	\, \e^{-s \wc{L}_{\vec 0}}\, (1-Q_{\vec 0}) \,   \partial_i {\Phi}_{\vec 0} }\right]\, \di s;  \label{eq:N2}\\
	N_3=&\int_{ 0}^{t}\ipc{ \partial_i\partial_j\,  \wc{L}  ^\dagger_{\vec 0}\, {\varphi}_{\vec 0}\ }{  
		\, \e^{-s \wc{L}_{\vec 0}}\, (1-Q_{\vec 0}) \,   {\Phi}_{\vec 0} }\, \di s; \label{eq:N3}
	\end{align}
	\begin{align}
	N_4=-\int_{ 0}^{t}\int_{ 0}^{s} & \left[\ipc{\, \partial_i \wc{L} ^\dagger_{\vec 0}\, {\varphi}_{\vec 0}\, }
	{ \e^{-(s-r) \wc{L}_{\vec 0}}(1-Q_{\vec 0}) \partial_j \wc{L} _{\vec 0}\, \e^{-r \wc{L}_{\vec 0}}\, (1-Q_{\vec 0}){\Phi}_{\vec 0} } \right .  ;\\
	& \left . +\ipc{\, \partial_j \wc{L} ^\dagger_{\vec 0}\, {\varphi}_{\vec 0}\, }
	{ \e^{-(s-r) \wc{L}_{\vec 0}}(1-Q_{\vec 0}) \partial_i \wc{L} _{\vec 0}\, \e^{-r \wc{L}_{\vec 0}}\, (1-Q_{\vec 0}){\Phi}_{\vec 0} }\right]\, \di r\, \di s \label{eq:N4}\\
	N_5=-\int_{ 0}^{t}\int_{ 0}^{s} &\left[\ipc{\, \partial_i \wc{L} ^\dagger_{\vec 0}\, {\varphi}_{\vec 0}\, }
	{ \e^{-(s-r) \wc{L}_{\vec 0}}(1-Q_{\vec 0}) \partial_j \wc{L} _{\vec 0}\, Q_{\vec 0} {\Phi}_{\vec 0} } \right . ; \\
	& \left . +\ipc{\, \partial_j \wc{L} ^\dagger_{\vec 0}\, {\varphi}_{\vec 0}\, }
	{ \e^{-(s-r) \wc{L}_{\vec 0}}(1-Q_{\vec 0}) \partial_i \wc{L} _{\vec 0} \, Q_{\vec 0}{\Phi}_{\vec 0} }\right]\, \di r\, \di s . \label{eq:N5}
	\end{align}
\end{lem}

Combining the above decomposition and the contraction property of $\e^{-t \wc{L}_{\vec 0}}$ in Lemma \ref{lem:L0dynamics}, we have the following convergence of $N_n$, which implies Lemma \ref{lem:lim-Mij} immediately.
\begin{lem}\label{lem:lim-N} Let $M_{i,j}=\sum\limits_{n=1}\limits^5N_n$ be given as in Lemma \ref{lem:M=N1-5}. Then 
 \begin{align}
		& \lim_{t\to\infty}\frac{1}{t}\abs{N_n}=0,\ n=1,\cdots,4. \label{eq:M12-lim}\\
	 &\lim_{t\to\infty}\frac{1}{t}N_5	= \ipc{\,  \partial_i  \wc{K}_{\vec 0}\, {\varphi}_{\vec 0}\, } 
		{  P_2\mc{J}^{-1}   \partial_j  \wc{K}_{\vec 0}\, {\varphi}_{\vec 0}} + \ipc{\,  \partial_j  \wc{K}_{\vec 0}\, {\varphi}_{\vec 0}\, } 
		{  P_2\mc{J}^{-1}   \partial_i  \wc{K}_{\vec 0}\, {\varphi}_{\vec 0}} . \label{eq:M3-lim} 
		\end{align}
\end{lem}

\begin{proof}
	
\noindent {\bf Case $n=1$:} Note that $ \partial_i \partial_j   {\Phi}_{\vec 0}=\sqrt{\otimes \vec p}\, \left .  \partial_i \partial_j   \wh{\rho}_{0;{\vec k}}\right|_{\vec k=\vec 0 }\otimes\id$.	Direct computation by \eqref{eq:rho-hat} shows 
\begin{align}
\pi_{\sigma}\, \partial_i \partial_j   \wh{\rho}_{0;{\vec 0}}(x) =-
\sum\limits_{n\in {\vec p}\Z^d+\sigma} n_in_j \psi_0(x-n) \overline{\psi_0(-n)}. \label{eq:rho''}
\end{align}

Therefore, by \eqref{eq:N1}
\begin{align*}
\abs{N_1}
=\abs{\ipc{ {\delta_{\vec 0}}\otimes \overrightarrow{1}\otimes\id\ }{ \, 	 \partial_i \partial_j   \wh{\rho}_{0;{\vec 0}}\otimes\id }}
=&\left|\sum\limits_{n\in \Z^d} n_in_j |\psi_0(n)|^2 \right|
\le \sum\limits_{n\in \Z^d} |n|^2 |\psi_0(n)|^2.
\end{align*}
Clearly, $\abs{N_1}$ is uniformly bounded in $t$ by \eqref{eq:initial}, which implies $\lim\limits_{t\to\infty}\frac{1}{t}|N_1(t)|=0.$\\

\noindent{\bf Case $n=2$:}  By \eqref{eq:rho-hat} and the same computation as in \eqref{eq:rho''}, we have $\partial_j {\Phi}_{\vec 0}=\left . \partial_j \wh{\rho}_{0;{\vec k}} \right|_{\vec k=\vec 0 }\otimes\id$ with 
\begin{align}
\pi_\sigma\,	\partial_j\,  \wh{\rho}_{0;{\vec 0}}(x) =-\im 
\sum\limits_{n\in {\vec p}\Z^d+\sigma} n_j\, \psi_0(x-n) \overline{\psi_0(-n)}.
\label{eq:rho'}
\end{align}
By \eqref{eq:initial}, \eqref{eq:initial-xl1} and direct computation, we obtain
	\begin{align*}
\norm{\partial_j \wh{\rho}_{0;\vec 0}}^2_{\ell^2(\Z^d;\C^{\otimes\vec{p}})}
	\le & \sum_{\sigma\in \Z_{\vec p} }\sum_{x\in\Z^d}\left|{\sum_{ n\in {\vec p}\Z^d+\sigma }  n_j\, \psi_0(x-n) \overline{\psi_0(-n)}}\right|^2 \nonumber\\
	\le &  \norm{\psi_0}^2_{\ell^2}\sum_{\sigma\in \Z_{\vec p} }\left(\sum_{ n\in {\vec p}\Z^d+\sigma }|n|\left| {\psi_0(-n)}\right|\right)^{2} \nonumber\\
	\le &  \norm{\psi_0}^2_{\ell^2}\cdot \norm{x\psi_0}^2_{\ell^1}<\infty. 
	\end{align*}
 By Lemma \ref{lem:K'}, 
	\begin{align*}
	\norm{\partial_j \wc{L} _{\vec 0}}_{L^2(\wh{M};\C^{\otimes\vec{p}})}=\norm{ \partial_j  \wc{K}_{\vec 0}}_{L^2(\wh{M};\C^{\otimes\vec{p}})}\le \|\wh h'\|_{\infty}. 
	\end{align*}
By Lemma \ref{lem:L0dynamics}, we have 
\begin{align*}
   &\int_{ 0}^{t}\abs{\ipc{ \partial_i \wc{L} ^\dagger_{\vec 0}{\varphi}_{\vec 0}\ }{\, \e^{-s \wc{L}_{\vec 0}}\, (1-Q_{\vec 0}) \,   \partial_j {\Phi}_{\vec 0} }}\, \di s \\
\le &   \norm{ \partial_i  \wc{K}_{\vec 0}{\varphi}_{\vec 0}} \cdot \norm{\partial_j {\Phi}_{\vec 0}} \cdot 
C_{\epsilon} 
\int_{ 0}^{t}\e^{-s ({g } -\epsilon)} \, \di s  \\
\le &  \|\wh h'\|_{\infty}\cdot\sqrt{\otimes \vec p }  \cdot \norm{\partial_j \wh{\rho}_{0;\vec 0}} 
 \cdot \frac{C_{\epsilon}}{{g }-\epsilon}<\infty.
\end{align*}
Therefore, 
$\lim\limits_{t\to\infty}\frac{1}{t}|N_2(t)|=0$. \\

\noindent {\bf Case $n=3$:} $N_3$ can be estimated exact in the same way as $N_2$. Again by Lemma \ref{lem:K'}, we have 
\begin{align*}
\norm{ \partial_i \partial_j   \wc{L}_{\vec 0}}_{L^2(\wh{M};\C^{\otimes\vec{p}})}=\norm{ \partial_i \partial_j   \wc{K}_{\vec 0}}_{L^2(\wh{M};\C^{\otimes\vec{p}})}\le \|\wh{h}''\|_{\infty}.
\end{align*}
By Lemma \ref{lem:L0dynamics}, we have 
\begin{align*}
\sup_t|N_{3}(t)|\le&  \sup_t\int_{ 0}^{t}\abs{\ipc{ \partial_i \partial_j  \, \wc{L} ^\dagger_{\vec 0}{\varphi}_{\vec 0}\ }{\, \e^{-s \wc{L}_{\vec 0}}\, (1-Q_{\vec 0}) \,   {\Phi}_{\vec 0} }}\, \di s\\
\le &  \|\wh{h}''\|_{\infty}\cdot \sqrt{\otimes \vec p } \cdot \norm{\wh{\rho}_{0;\vec 0}}_{\ell^2} \cdot 
 \frac{C_{\epsilon}}{{g }-\epsilon}<\infty ,
\end{align*}
which gives $\lim\limits_{t\to\infty}\frac{1}{t}\abs{N_{3}(t)}=0$. \\

\noindent {\bf Case $n=4$:} $N_4$ can be estimated by applying  Lemma \ref{lem:L0dynamics} twice:
\begin{align*}
&\sup_t\int_{0}^{t} \int_{0}^{s} \ipc{\, \partial_i \wc{L} ^\dagger_{\vec 0}\, {\varphi}_{\vec 0}\, }
{   \e^{-(s-r) \wc{L}_{\vec 0} }\, (1-Q_{\vec 0})\partial_j \wc{L} _{0}\, \e^{-r \wc{L}_{\vec 0}}(1-Q_{\vec 0})\,{\Phi}_{0}} 
\di r\, \di s \\
\le &   \norm{ \partial_i  \wc{K}_{\vec 0}{\varphi}_{\vec 0}} \cdot \norm{\partial_j \wc{L} _{0}}\cdot \norm{{\Phi}_{0}} \cdot 
C^2_{\epsilon} \sup_t
\int_{0}^{t}\int_{0}^{s} \e^{-(s-r) ({g } -\epsilon)}\,\e^{-r ({g } -\epsilon)} \, \di r\,\di s  \\
\le & \|\wh{h}'\|^2_{\infty}\cdot \sqrt{\otimes \vec p } \cdot \norm{\wh{\rho}_{0;\vec 0}}_{\ell^2} \cdot 
C^2_{\epsilon} \cdot \left(\frac{1}{{g }-\epsilon}
+\frac{1}{({g }-\epsilon)^2}\right)<\infty,
\end{align*}
and thus, $\lim\limits_{t\to\infty}\frac{1}{t}\abs{N_{4}(t)}=0$.  \\

\noindent {\bf Case $n=5$:} It remains to estimate $\frac{1}{t}N_5$. Recall we obtained  $Q_{\vec 0}{\Phi}_{0}= {\varphi}_{\vec 0}$ in \eqref{eq:Q0Phi_0}. This 
\begin{align*}
&\frac{1}{t}  \int_{0}^{t} \int_{0}^{s} \ipc{\, \partial_i \wc{L} ^\dagger_{\vec 0}\, {\varphi}_{\vec 0}\, }
{  \e^{-(s-r) \wc{L}_{\vec 0} }\, (1-Q_{\vec 0})\partial_j \wc{L} _{0}\, \e^{-r \wc{L}_{\vec 0}}\, Q_{\vec 0}\,{\Phi}_{0}} 
\di r\, \di s \\
=&-    \ipc{\,  \partial_i  \wc{K}_{\vec 0}\, {\varphi}_{\vec 0}\, } 
{  \left(\frac{1}{t} \int_{0}^{t}\int_{0}^{s} \e^{-(s-r)  \wc{L}_{\vec 0} }\,\di r\, \di s \right)   \partial_j  \wc{K}_{\vec 0}\, {\varphi}_{\vec 0}}, 
\end{align*}
since $\partial_j  \wc{K}_{\vec 0}{\varphi}_{\vec 0}\in  \ran (1-Q_{\vec 0})$, $(1-Q_{\vec 0})\partial_j  \wc{K}_{\vec 0}{\varphi}_{\vec 0}=\partial_j  \wc{K}_{\vec 0}{\varphi}_{\vec 0}$. 

Since $\Re \wc{L}_{\vec 0}\ge 0$, by a standard contour integral argument, the following formula was obtained in \cite{KS2009, Schenker:2015}
	\begin{align}
	\lim_{t\to\infty}\frac{1}{t} \int_{0}^{t}\int_{0}^{s} \Pi_2\e^{-(s-r)  \wc{L}_{\vec 0} }\Pi_2\,\di r\, \di s= \Pi_2\left((1-Q_{\vec 0})\wc{L}_{\vec 0}(1-Q_{\vec 0})\right)^{-1}\Pi_2=\Pi_2\mc{J}^{-1}\Pi_2, 
	\end{align}
	where $\mc{J}^{-1}$ is as in Lemma \ref{lem:L0inverse} .
Recall that $\partial_i  \wc{K}_{\vec 0}\in {\rm Ran}(\Pi_2)\subseteq {\rm Ran}(P_2)$. Thus
\begin{align*}
\lim_{t\to\infty}\frac{1}{t}N_5
=&   \ipc{\,  \partial_i  \wc{K}_{\vec 0}\, {\varphi}_{\vec 0}\, } 
{  \Pi_2\mc{J}^{-1}\Pi_2  \partial_j  \wc{K}_{\vec 0}\, {\varphi}_{\vec 0}}+   \ipc{\,  \partial_j  \wc{K}_{\vec 0}\, {\varphi}_{\vec 0}\, } 
{\Pi_2\mc{J}^{-1}\Pi_2  \partial_i  \wc{K}_{\vec 0}\, {\varphi}_{\vec 0}} \\
=&   \ipc{\,  \partial_i  \wc{K}_{\vec 0}\, {\varphi}_{\vec 0}\, } 
{  P_2\mc{J}^{-1}  \partial_j  \wc{K}_{\vec 0}\, {\varphi}_{\vec 0}}+   \ipc{\,  \partial_j  \wc{K}_{\vec 0}\, {\varphi}_{\vec 0}\, } 
{ P_2\mc{J}^{-1}  \partial_i  \wc{K}_{\vec 0}\, {\varphi}_{\vec 0}} \\
=&\partial_i\partial_jE(\vec 0),
\end{align*}
where the last line follows from the formula of $\partial_i\partial_jE(\vec 0)$ in \eqref{eq:Eij}. 
\end{proof}

\subsection{Limiting behavior of $\vec{D}(\lambda)$ for small $\lambda$}
The following lemma can be found in \cite{Schenker:2015}. It will be the main tool for us to study the asymptotic behavior of $\vec{D}(\lambda)$. 
\begin{lem}[Lemma D.1, \cite{Schenker:2015}]\label{lem:limres} Let $A$ and $R$ be bounded operators on a Hilbert space $\mc{H}$.  If $A $ is normal, $\Re A \ge 0$ and $\Re R \ge c > 0$, then for any $\phi,\psi\in \mc{H}$, 
	$$\lim_{\eta \rightarrow 0} \ipc{\phi}{\left ( {\eta}^{-1} A + R \right )^{-1} \psi}_{\mc{H}} \ = \ \ipc{\Pi \phi}{ \left ( \Pi R \Pi \right )^{-1} \Pi \psi}_{\ran \Pi}$$
	where $\Pi=$ projection onto the kernel of $A$. 
\end{lem}
\begin{remark}\label{rem:limres}
	A similar statement holds for a family of bounded operators $R_{\eta}$ such that $\Re R_{\eta}\ge c>0$ and $\lim_{\eta\to0}R_{\eta}=R_0$ in the strong operator topology and $R_0\ge c>0$, i.e., 
	$$\lim_{\eta \rightarrow 0} \ipc{\phi}{\left ( {\eta}^{-1} A + R_{\eta} \right )^{-1} \psi}_{\mc{H}} \ = \ \ipc{\Pi \phi}{ \left ( \Pi R_0 \Pi \right )^{-1} \Pi \psi}_{\ran \Pi}.$$
	 
\end{remark}
In view of Lemma \ref{lem:L0inverse}, we want to have the block form of the above lemma. 
\begin{lem}\label{lem:Ablock}
	Let $A$  be a bounded self-adjoint operator on a Hilbert space $\mc{H}=\mc{H}_1\oplus \mc{H}_2$ with the following block form:
	\begin{align}
	A= \begin{pmatrix}
	0 & A_{2} \\
	A_{2}^{\dagger} & A_{3}
	\end{pmatrix}, \ \ A_{3}^{\dagger}=A_{3}.\label{eq:A}
	\end{align}
	Let ${\Pi}=$ projection onto the kernel of $A$, $\Pi_2=$ projection onto the kernel of $A_{2}$ and ${\wt \Pi=}$ projection onto the kernel of $\Pi_2A_{3}\Pi_2$. For any $\varphi=\Pi_2\varphi$,  
	\begin{align*}
	\ {\Pi} \varphi=0\ \ {\textrm {if and only if}}\ \ {{\wt \Pi} \varphi}=0. 
	\end{align*}	

\end{lem}
\begin{proof}
	For any $\varphi \in\mc{H}$,  direct application of Lemma \ref{lem:limres} to $I+\im\, \eta^{-1}  A$ gives 
	\begin{align}\label{eq:lim-1}
	\lim_{\eta \to 0} \ipc{\varphi}{(I+\im\,\eta^{-1}A)^{-1}\varphi}= \ipc{{\Pi}\varphi}{({\Pi}\,I\,{\Pi})^{-1}{\Pi}\varphi}=\norm{{\Pi}\varphi}^2. 
	\end{align}
Let $P_1,P_2$ be the projection onto $\mc{H}_1, \mc{H}_2$ correspondingly and consider  $\varphi\in\ran(\Pi_2)$. By the block form of $A$ and Schur's formula, we have
	\begin{align}
	\ipc{\varphi}{(I+\im\,\eta^{-1}A)^{-1}\varphi}
	=&\ipc{\varphi}{(P_2+\im \eta^{-1}A_{3}+\eta^{-2}A_{2}^{\dagger}A_{2} )^{-1}\, \varphi}\nonumber\\
	=&\ipc{\varphi}{\Pi_2(P_2+\im\eta^{-1}A_{3}+\eta^{-2}A_{2}^{\dagger}A_{2} )^{-1}\Pi_2\, \varphi}. \label{eq:Schur-1}
	\end{align}
If we apply Schur's formula one more time with respect to the decomposition $\mc{H}_2=\ran(\Pi_2)\oplus \ran (\Pi_2^{\perp})$ and notice that $\Pi_2A_{2}^{\dagger}=A_{2}\Pi_2=0$, then we have 
	\begin{align*}
	\ipc{\varphi}{(I+\im\,\eta^{-1}A)^{-1}\varphi}
=\ipc{\varphi}
{\left(\,\im\eta^{-1}\Pi_2A_{3}\Pi_2+\,\Pi_2+\wt{A}\, \right)^{-1}\, \varphi}
	\end{align*}
	where $\wt{A}=
	\Pi_2A_{3}\Pi_2^{\perp}
	\left(\eta^2\Pi_2^{\perp}+\im	\eta\Pi_2^{\perp}A_{3}\Pi_2^{\perp}+\Pi_2^{\perp}A_{2}^{\dagger}A_{2}\Pi_2^{\perp} \right)^{-1}
	\Pi_2^{\perp}A_{3}\Pi_2$. By \eqref{eq:A}, $\Re \wt{A}\ge0$ on $\ran(\Pi_2)$, which implies $\Re (\Pi_2+\wt{A})\ge1>0$ on  $\ran(\Pi_2)$. 

	Therefore, by Lemma \ref{lem:limres} and Remark \ref{rem:limres}, we have that 
	\begin{align}
	&\lim_{\eta \to 0} \ipc{\varphi}{(I+\im\,\eta^{-1}A)^{-1}\varphi}\nonumber \\
	=&\lim_{\eta \to 0} \ipc{\varphi}
	{\left(\,\im\eta^{-1}\Pi_2A_{3}\Pi_2+\,\Pi_2+\wt{A}\, \right)^{-1}\, \varphi} \nonumber\\
	=& \ipc{{{\wt \Pi} \varphi}}
	{\left(\,\wt\Pi +\wt\Pi \,
		\Pi_2A_{3}\Pi_2^{\perp}
		\left(\Pi_2^{\perp}A_{2}^{\dagger}A_{2}\Pi_2^{\perp} \right)^{-1}
		\Pi_2^{\perp}A_{3}\Pi_2\wt\Pi\, \right)^{-1}\, {{\wt \Pi} \varphi}}, \label{eq:lim-2}
	\end{align}
	where ${\wt \Pi=}$ projection onto the kernel of $\Pi_2A_{3}\Pi_2$.

	Putting \eqref{eq:lim-1} and \eqref{eq:lim-2} together, we have that 
	$$\ipc{{{\wt \Pi} \varphi}}
	{\left(\,\wt\Pi+\wt\Pi \,
		\Pi_2A_{3}\Pi_2^{\perp}
		\left(\Pi_2^{\perp}A_{2}^{\dagger}A_{2}\Pi_2^{\perp} \right)^{-1}
		\Pi_2^{\perp}A_{3}\Pi_2\,\wt\Pi \right)^{-1}\, {{\wt \Pi} \varphi}}=\norm{{\Pi}\varphi}^2,$$
which completes the proof of Lemma \ref{lem:Ablock}. 
\end{proof}

Now we can proceed to prove \eqref{eq:genAsym} in Theorem \ref{thm:genCLT}. As showed in eq. \eqref{eq:Eij} and Lemma \ref{lem:lim-Mij}, the diffusion matrix $\vec{D}(\lambda)$ is independent of the initial condition $\psi_0\in\ell^2(\Z^d)$. To study the asymptotic behavior of $\vec D(\lambda)$, it is enough to consider $
\psi_0(x)=\delta_{\vec 0}$, where we assume the ballistic motion holds in \eqref{eq:balli-gen}. 
\begin{proof}[Proof of \eqref{eq:genAsym}]
	 We are going  to apply Lemma \ref{lem:Ablock} to $A$ acting on $\wc{H}_1\oplus \wc{H}_2$ given by:
\begin{align}
A= P(\wc{K}_{\vec 0}+\wc{U})P= \begin{pmatrix}
0 & P_1\wc{K}_{\vec 0}P_2 \\
P_2\wc{K}_{\vec 0}P_1 & P_2(\wc{K}_{\vec 0}+\wc{U} )P_2
\end{pmatrix},
\end{align}
where $\wc{H}_i,P_i,i=1,2$ are as in \eqref{eq:directsum} and $P=P_1+P_2$.
 Let ${\Pi}=$ projection onto the kernel of $P(\wc{K}_{\vec 0}+\wc{U})P$, $\Pi_2=$ projection onto the kernel of $P_1\wc{K}_{\vec 0}P_2$ and ${\wt \Pi=}$ projection onto the kernel of $\Pi_2 (\wc{K}_{\vec 0}+\wc{U} )\Pi_2$. 
 
Let $\wt{\phi}_j=\partial_j\wc{K}_{\vec 0}\varphi_{\vec 0},j=1\cdots,d$, which are given as in \eqref{eq:K0'-delta0}. Recall that $\wt{\phi}_j\in {\rm Ker}(P_1\wc{K}_{\vec 0}P_2)$, therefore $\wt{\phi}_j=\Pi_2\wt{\phi}_j$. Let $M_{i,j}$ be as in \eqref{eq:Mij} and $\wh{\rho}_{0;{\vec k}}$ be as in \eqref{eq:rho-hat}.
By the decomposition in Lemma \eqref{lem:M=N1-5} at $\lambda=0$, one can  check that 
\footnote{This formula was obtained in \cite{Schenker:2015}, Sec. 4.7, where there is no error term $O(\eta^2)$. In \cite{Schenker:2015}, the choice $\psi_0=\delta_{\vec 0}$ implies that  $M_{j,j}(0)=0$ and $\wh{\rho}_{0;{\vec 0}}=\delta_{\vec 0}\otimes \ora{1}$ and the proof is relatively simple.  In the general $\vec p$-periodic case, the initial condition $\delta_{\vec 0}$ no longer provides the simplified expressions of $M_{j,j}(0)$ and $\wh{\rho}_{0;{\vec 0}}$. We need the correction term for small $\eta$. 
The proof for the general case is essentially based on the same strategy for Lemma \ref{lem:lim-N}; we omit the details here.  }
\begin{align}\label{eq:wtMjj}
2\ipc{\wt{\phi}_j}{\left(P+\eta^{-1}\im(\wc{K}_{\vec 0}+\wc{U})\right)^{-1}\wt{\phi}_j}=\eta^3\int_0^{\infty}e^{-\eta\,t}\,  M_{j,j}(t)\, \di t\,+ O(\eta^2).
\end{align}


When $\lambda=0$, $\wc{L}_{\vec 0}=\im(\wc{K}_{\vec 0}+\wc{U})$ is the unperturbed periodic operator on $\ell^2({\Z^d;\C^{\otimes \vec p}})$. Setting $\eta=2T^{-1}$ in \eqref{eq:balli-gen}, there is a $c>0$ such that for all $j$ and $\eta$ small,
\begin{align}\label{eq:ballistic}
\eta^3\int_0^{\infty}e^{-\eta\,t}\,  M_{j,j}(t)\, \di t=
\frac{8}{T^3}\int_0^{\infty}\,e^{-\frac{2t}{T}}\,\sum_{x\in \Z^d} x_j^2\, \Ev{\abs{\psi_t(x)}^2}\, \di t\,\ge c>0.
\end{align}
Put \eqref{eq:lim-1}, \eqref{eq:wtMjj} and \eqref{eq:ballistic} together, we have 
\begin{align*}
\norm{\Pi\wt{\phi}_j}^2=\lim_{\eta\to 0}\ipc{\wt{\phi}_j}{\left(P+\eta^{-1}\im(\wc{K}_{\vec 0}+\wc{U})\right)^{-1}\wt{\phi}_j} >0.
\end{align*}
Therefore, ${{ \Pi} \wt\phi}_j\neq0$ and Lemma \ref{lem:Ablock} implies that ${{\wt \Pi} \wt\phi}_j\neq0$. 

Recall that $\wc{L}_{\vec{0}} =  \im \wc{K}_{{\vec 0}} + \im \wc{U} + \im {\lambda} \wc{V}  + B$ and
 $\Gamma_2=P_2\left(\im \wc{K}_{\vec 0}+\im \wc{U}+{\lambda}^2\wc{V}\,(P_3\wc{L}_{{\vec 0}}P_3)^{-1}\, \wc{V}\right)P_2$  as in \eqref{eq:Gamma2}. Let 
 $R_{\lambda}=\Pi_2\wc{V}\left(P_3\wc{L}_{\vec0}P_3\right)^{-1}\wc{V}\Pi_2$ and $R_0=\Pi_2\wc{V}\left(P_3(\im \wc{K}_{\vec 0}+\im \wc{U})P_3\right)^{-1}\wc{V}\Pi_2$. Then $\Pi_2\Gamma_2\Pi_2=\im\,\Pi_2P_2(\wc{K}_{\vec 0}+\wc{U} )P_2\Pi_2+\lambda^2R_{\lambda}$ and $\lim_{\lambda\to0}R_{\lambda}=R_0$ (in the strong operator topology). Applying Lemma \ref{lem:limres} (and Remark \ref{rem:limres})  to $\Pi_2\Gamma_2\Pi_2$ on $\ran(\Pi_2)$, we obtain that, for any $1\le i,j\le d$,
\begin{align*}
\lim_{\lambda\to0}\lambda^2\ipc{\wt{\phi}_i}{\left(\Pi_2\Gamma_2\Pi_2\right)^{-1}\,\wt{\phi}_j}=&
\lim_{\lambda\to0}\ipc{\wt{\phi}_i}{\left(\im\,\lambda^{-2}\Pi_2(\wc{K}_{\vec 0}+\wc{U} )\Pi_2+R_{\lambda}\right)^{-1}\,\wt{\phi}_j}\\
=&\ipc{{{\wt \Pi} \wt\phi}_i}{\left(\wt\Pi R_{0}\wt\Pi\right)^{-1}\,\wt \Pi\wt{\phi}_j}. 
\end{align*}
In particular, $\lim_{\lambda\to0}\lambda^2\ipc{\wt{\phi}_j}{\left(\Pi_2\Gamma_2\Pi_2\right)^{-1}\,\wt{\phi}_j}
=\ipc{{{\wt \Pi} \wt\phi}_j}{\left(\wt\Pi R_{0}\wt\Pi\right)^{-1}\,{{\wt \Pi} \wt\phi}_j}>0$. 

By Lemma \ref{lem:L0inverse}  and \eqref{eq:Eij}, we have 
\begin{align*}
\lim_{\lambda\to0}\lambda^2\partial_i\partial_jE(\vec 0)=\ipc{{{\wt \Pi} \wt\phi}_j}{\left(\wt\Pi R_{0}\wt\Pi\right)^{-1}\,{{\wt \Pi} \wt\phi}_i}+\ipc{{{\wt \Pi} \wt\phi}_i}{\left(\wt\Pi R_{0}\wt\Pi\right)^{-1}\,{{\wt \Pi} \wt\phi}_j}=:\vec D^0_{ij}. 
\end{align*}
Let $\vec D^0:=(\vec D^0_{ij})_{d\times d}$. Then $\lim_{\lambda\to0}\lambda^2 \vec D=\vec D^0$ and $\ipc{\vec k}{\vec D^0 \vec k}>0$ for any $\vec0\neq \vec k\in\R^d$ by the same argument for $\vec D$. As a consequence, 
\begin{align*}
\lim_{\lambda\to0}\lambda^2 \tr \vec D=\tr \vec D^0>0. 
\end{align*} 
This completes the proof of Theorem \ref{thm:asym}. 
\end{proof}
\appendix
\section{Decomposition of the second moments and the proof of Lemma \ref{lem:M=N1-5}}\label{app:A}

The following facts will be used to simplify the expression of the second order partial derivative. Note that $
\wc{L}_{\vec 0} \, \varphi_{\vec 0}
=\wc{L}^{\dagger}_{\vec 0} \, \varphi_{\vec 0}\ = \ 0,
$ implies that $\e^{-t \wc{L}_{\vec 0}}$ and $\e^{-t \wc{L}^\dagger_{\vec 0}} $ act trivially on $\varphi_{\vec 0}$ for any $t$, i.e.,
\begin{align}
\e^{-t \wc{L}_{\vec 0}}\varphi_{\vec 0}=\e^{-t \wc{L}^\dagger_{\vec 0}}\varphi_{\vec 0}=\varphi_{\vec 0} \label{eq:exptL}
\end{align}
and 
\begin{align}
\e^{-t \wc{L}_{\vec 0}}Q_{\vec 0}=\e^{-t \wc{L}^\dagger_{\vec 0}}Q_{\vec 0}=Q_{\vec 0}. \label{eq:exptL-Q0}
\end{align}
On the other hand, recall the formula for differentiating a semi-group,
\begin{equation} \label{eq:expLk-deriv}
\partial_j \left(\e^{-t \wc{L}_{\vec k}} \right)   \ = \  -\int_{ 0}^{t} \e^{-(t-s) \wc{L}_{\vec k} }\, \partial_j \wc{L} _{\vec k}\, \e^{-s \wc{L}_{\vec k}}\, \di s. 
\end{equation}	
By \eqref{eq:Kk} and \eqref{def:Lk}, we have $
\partial_j \wc{L} _{\vec 0}=\im \,	 \partial_j  \wc{K}_{\vec 0}=-\partial_j\,\wc{L}^{\dagger}_{\vec 0}.
$
Because $\partial_j  \wc{K}_{\vec 0}$ maps $\wc{H}_0\oplus\wc{H}_1$ to $\wc{H}_2$, we also have that 
\begin{align} \label{eq:Q0L'}
Q_{\vec 0}\partial_j \wc{L} _{\vec 0}=Q_{\vec 0}\partial_j \wc{L} ^{\dagger}_{\vec 0}=0;
\end{align}
\begin{align}
\partial_i	\partial_j  \wc{L}  _{\vec 0}=\im 	 \partial_i	\partial_j  \wc{K}_{\vec 0}=-\partial_i	\partial_j  \wc{L}  _{\vec 0}^{\dagger} \ {\rm and }\ 
Q_{\vec 0}\partial_i	\partial_j  \wc{L}  _{\vec 0}=Q_{\vec 0}\partial_i	\partial_j \wc{L}  ^{\dagger}_{\vec 0}=0.
\end{align}
	Direct computation from \eqref{eq:Mij-1} gives
	\begin{align}
	M_{i,j} (t) =& - \left .  \partial_{i} \partial_{j} \ipc{{\varphi}_{\vec 0} }{\e^{-t \wc{L}_{{\vec k}}} \,{\Phi}_{\vec k} }\right |_{\vec{k}=\vec 0} \nonumber\\
	=& -\ipc{ {\varphi}_{\vec 0}\ }{\e^{-t \wc{L}_{\vec 0}} \, \partial_{i} \partial_{j} {\Phi}_{\vec 0} } \label{eq:N1-pf}\\
	&-\ipc{ {\varphi}_{\vec 0}\ }{\left(\partial_{i}\e^{-t \wc{L}_{\vec 0}}\right)_{|{\vec k}=\vec 0} \,  \partial_{j} {\Phi}_{\vec 0} }-\ipc{ {\varphi}_{\vec 0}\ }{\left(\partial_{j}\e^{-t \wc{L}_{\vec 0}}\right)_{|{\vec k}=\vec 0} \,  \partial_{i} {\Phi}_{\vec 0} }\label{eq:N2-pf}\\
	&-\ipc{ {\varphi}_{\vec 0}\ }{\left(\partial_{i}\partial_{j}\e^{-t \wc{L}_{\vec 0}}\right)_{|{\vec k}=\vec 0} \,   {\Phi}_{\vec 0} } . \label{eq:N345-pf}
	\end{align}
	Clearly, \eqref{eq:N1-pf} gives the expression for $N_1$ in \eqref{eq:N1}. Now let's proceed to simplify the expression in \eqref{eq:N2-pf}. By the differential formula \eqref{eq:expLk-deriv}, we obtain
	\begin{align*}
\ipc{ {\varphi}_{\vec 0}\ }{ \left(\partial_i\, \e^{-t \wc{L}_{\vec k}} \right) _{|{{\vec k}=0}}\, \partial_j {\Phi}_{\vec 0} }
	=&   \ipc{ {\varphi}_{\vec 0}\ }{\    
		\left(-\int_{ 0}^{t} \e^{-(t-s) \wc{L}_{\vec 0} }\, \partial_i  \wc{L} _{\vec 0}\, \e^{-s \wc{L}_{\vec 0}}\, \di s\right)
		\,  \partial_j  {\Phi}_{\vec 0} }  \\
	=& 
	-  \int_{ 0}^{t}\ipc{ {\varphi}_{\vec 0}\ }{\    
		\partial_i  \wc{L} _{\vec 0}\, \e^{-s \wc{L}_{\vec 0}}\, (1-Q_{\vec 0}) \,   \partial_j  {\Phi}_{\vec 0} }\, \di s  ,
	\end{align*}
	where we use the fact by \eqref{eq:exptL-Q0} that
$
\ipc{ {\varphi}_{\vec 0}\ }{\    
	\partial_i  \wc{L} _{\vec 0}\, \e^{-s \wc{L}_{\vec 0}}\, Q_{\vec 0}
	\,  \partial_j  {\Phi}_{\vec 0} }
=0. 
$ This gives the expression for $N_2$ in \eqref{eq:N2}. 


	Simplifying \eqref{eq:N345-pf} requires applying \eqref{eq:expLk-deriv} twice. Differentiating \eqref{eq:expLk-deriv} again yields, 
	\begin{align*}
	\qquad \left .\partial_i\partial_j \left(\e^{-t \wc{L}_{\vec k}} \right)\right|_{\vec k=\vec 0 } 
	=&-\int_{ 0}^{t} \e^{-(t-s) \wc{L}_{\vec 0} }\, \partial_i \partial_j   \wc{L}  _{\vec 0}\, \e^{-s \wc{L}_{\vec 0}}\, \di s \\
	&+\int_{ 0}^{t} \left(\int_{ 0}^{t-s} \e^{-(t-s-r) \wc{L}_{\vec 0} }\, \partial_i \wc{L} _{\vec 0}\, \e^{-r \wc{L}_{\vec 0}}\, \di r \right)\, \partial_j \wc{L} _{\vec 0}\, \e^{-s \wc{L}_{\vec 0}}\, 
	\di s\\
	&+\int_{ 0}^{t} \e^{-(t-s) \wc{L}_{\vec 0} }\, \partial_j \wc{L} _{\vec 0}\, \left(\int_{ 0}^{s} \e^{-(s-r) \wc{L}_{\vec 0} }\, \partial_i \wc{L} _{\vec 0}\, \e^{-r \wc{L}_{\vec 0}}\, \di r\right)\, \di s .
	\end{align*}
Therefore, 
\begin{align}- \ipc{ {\varphi}_{\vec 0}\ }{   \left(\partial_i \partial_j  \e^{-t \wc{L}_{\vec k}} \right) _{|{\vec k=\vec 0 }}\, {\Phi}_{\vec 0} }
= \int_{ 0}^{t}\, \ipc{ {\varphi}_{\vec 0}\ } {\, \e^{-(t-s) \wc{L}_{\vec 0} }\,  \partial_i \partial_j  \,  \wc{L}  _{\vec 0}\, \e^{-s \wc{L}_{\vec 0}}\, {\Phi}_{\vec 0} }\, \di s \label{eq:N3-pf}\\
-  \int_{ 0}^{t}\int_{ 0}^{t-s} \ipc{ {\varphi}_{\vec 0}\, }
{\, \e^{-(t-s-r) \wc{L}_{\vec 0} }\, \partial_i \wc{L} _{\vec 0}\, \e^{-r \wc{L}_{\vec 0}} \partial_j \wc{L} _{\vec 0}\, \e^{-s \wc{L}_{\vec 0}}\, {\Phi}_{\vec 0} }  \di r\, \di s \label{eq:N45-pf-1} \\
-  \int_{ 0}^{t} \int_{ 0}^{s} \ipc{ {\varphi}_{\vec 0}\, }
{\, \e^{-(t-s) \wc{L}_{\vec 0} }\, \partial_j \wc{L} _{\vec 0}\,  \e^{-(s-r) \wc{L}_{\vec 0} }\, \partial_i \wc{L} _{\vec 0}\, \e^{-r \wc{L}_{\vec 0}}\,{\Phi}_{\vec 0}} 
\di r\, \di s . \label{eq:N45-pf-2}
\end{align}

The expression on the right hand side of \eqref{eq:N3-pf} leads to $N_3$ in \eqref{eq:N3} since
\begin{align*}\ipc{ {\varphi}_{\vec 0}\ } {\, \e^{-(t-s) \wc{L}_{\vec 0} }\,  \partial_i \partial_j  \,  \wc{L}  _{\vec 0}\, \e^{-s \wc{L}_{\vec 0}}\, {\Phi}_{\vec 0} }=
\ipc{  \partial_i \partial_j  \,  \wc{L}  ^\dagger_{\vec 0}{\varphi}_{\vec 0}\ }{  
		\, \e^{-s \wc{L}_{\vec 0}}\, (1-Q_{\vec 0}) \,   {\Phi}_{\vec 0} }.
	\end{align*}
	
Expressions for \eqref{eq:N45-pf-1} and \eqref{eq:N45-pf-2} follow from  \eqref{eq:exptL-Q0} and \eqref{eq:Q0L'} by direct computations. For  \eqref{eq:N45-pf-1} we have, 
\begin{align}
-  \int_{ 0}^{t}\int_{ 0}^{t-s}& \ipc{ {\varphi}_{\vec 0}\, }
{\, \e^{-(t-s-r) \wc{L}_{\vec 0} }\, \partial_i \wc{L} _{\vec 0}\, \e^{-r \wc{L}_{\vec 0}} \partial_j \wc{L} _{\vec 0}\, \e^{-s \wc{L}_{\vec 0}}\, {\Phi}_{\vec 0} }  \di r\, \di s  \nonumber \\
 =& -\int_{ 0}^{t}\int_{ 0}^{s} \left[\ipc{\, \partial_i \wc{L} ^\dagger_{\vec 0}\, {\varphi}_{\vec 0}\, }
{ \e^{-(s-r) \wc{L}_{\vec 0}}(1-Q_{\vec 0}) \partial_j \wc{L} _{\vec 0}\, \e^{-r \wc{L}_{\vec 0}}\, (1-Q_{\vec 0}){\Phi}_{\vec 0} } \right .  \label{eq:N4-1}\\
    &\left . +\ipc{\, \partial_i \wc{L} ^\dagger_{\vec 0}\, {\varphi}_{\vec 0}\, }
{ \e^{-(s-r) \wc{L}_{\vec 0}}(1-Q_{\vec 0}) \partial_j \wc{L} _{\vec 0}\, Q_{\vec 0} {\Phi}_{\vec 0} } 
\right]\, \di r\, \di s  . \label{eq:N5-1}
\end{align}
Similarily, for \eqref{eq:N45-pf-2}, 
\begin{align}
-  \int_{ 0}^{t} \int_{ 0}^{s}& \ipc{ {\varphi}_{\vec 0}\, }
{\, \e^{-(t-s) \wc{L}_{\vec 0} }\, \partial_j \wc{L} _{\vec 0}\,  \e^{-(s-r) \wc{L}_{\vec 0} }\, \partial_i \wc{L} _{\vec 0}\, \e^{-r \wc{L}_{\vec 0}}\,{\Phi}_{\vec 0}} 
\di r\, \di s \nonumber \\ 
=& -\int_{ 0}^{t}\int_{ 0}^{s} \left[\ipc{\, \partial_j \wc{L} ^\dagger_{\vec 0}\, {\varphi}_{\vec 0}\, }
{ \e^{-(s-r) \wc{L}_{\vec 0}}(1-Q_{\vec 0}) \partial_i \wc{L} _{\vec 0}\, \e^{-r \wc{L}_{\vec 0}}\, (1-Q_{\vec 0}){\Phi}_{\vec 0} }\right .  \label{eq:N4-2}\\
&\left . +\ipc{\, \partial_j \wc{L} ^\dagger_{\vec 0}\, {\varphi}_{\vec 0}\, }
{ \e^{-(s-r) \wc{L}_{\vec 0}}(1-Q_{\vec 0}) \partial_i \wc{L} _{\vec 0} \, Q_{\vec 0}{\Phi}_{\vec 0} }\right]\, \di r\, \di s . \label{eq:N5-2}
\end{align}
Clearly, 
\begin{align}
N_4=\eqref{eq:N4-1}+\eqref{eq:N4-2},\ \ N_5=\eqref{eq:N5-1}+\eqref{eq:N5-2}. 
\end{align}
This completes the proof of Lemma 5.2.

\section*{Acknowledgement}
The authors would like to thank Ilya Kachkovskiy for useful discussions on ballistic motion of periodic operators. Zak Tilocco and Jeffrey Schenker were supported by the National Science Foundation under Grant No. 1500386 and Grant No. 1411411. Shiwen Zhang was supported in part by NSF grant DMS-1600065, DMS-1758326, and by a post-doctoral fellowship from the MSU Institute for Mathematical and Theoretical Physics.


%
%
%
%
%
%
%
%
%

\end{document}